\documentclass[journal, onecolumn]{IEEEtran}
\usepackage{color, verbatim, hyperref}
\usepackage{amssymb}
\usepackage{amsmath, amsthm}
\usepackage{wrapfig}
\usepackage{epsfig}
\usepackage[usenames,dvipsnames]{xcolor}
\usepackage{mathdots}
\usepackage{tikz}
\usepackage{bbm}
\usepackage{listings, morefloats}
\usetikzlibrary{arrows}
\ifCLASSINFOpdf
\else
\fi
\begin{document}
\newtheorem{theorem}{Theorem}
\newtheorem{example}{Example}
\newtheorem{claim}{Claim}
\newtheorem{corollary}{Corollary}
\newtheorem{conjecture}{Conjecture}
\newtheorem{definition}{Definition}
\newtheorem{lemma}{Lemma}
\newtheorem {remark}{Remark}

%
\title{Critical Graphs in Index Coding}
%
%
%

\author{Mehrdad Tahmasbi, Amirbehshad Shahrasbi, and Amin Gohari
\\ {\it \small Electrical Engineering Department, Sharif University of Technology, Tehran, Iran.}\\
{\it \small emails: \{tahmasebi\_mehrdad, shahrasbi\}@ee.sharif.edu, aminzadeh@sharif.edu.}}
\maketitle

\begin{abstract}
In this paper we define critical graphs as minimal graphs that support a given set of rates for the index coding problem, and study them for both the one-shot and asymptotic setups. For the case of equal rates, we find the critical graph with minimum number of edges for both one-shot and asymptotic cases. For the general case of possibly distinct rates, we show that for one-shot and asymptotic linear index coding, as well as asymptotic non-linear index coding, each critical graph is a union of disjoint strongly connected subgraphs (USCS). On the other hand, we identify a non-USCS critical graph for a one-shot non-linear index coding problem. Next, we identify a few graph structures that are critical. We also generalize some of our results to the groupcast problem. In addition, we show that the capacity region of the index coding is additive for union of disjoint graphs. 
\end{abstract}


%
\IEEEpeerreviewmaketitle


\section{Introduction}
\IEEEPARstart{I}{ntroduced}  by Birk and Kol in \cite{YitzhakTomer}, index coding is the problem of transmitting a set of messages to a number of receivers via a public communication. Each receiver may also have some side information consisting of messages desired by some of the other receivers. This problem has been the subject of several recent studies (e.g. see \cite{SonHoangVitalyYeowMeng}-\cite{HamedViveckSyedAli}) In the most general form of the problem, each message can be desired by more than one destination. However the special case of each message being desired by exactly one receiver admits a graph theoretic representation in terms of directed graphs and thus has received particular attention. More specifically, if there are $m$ receivers, we can construct a graph with $m$ vertices. We draw a directed edge from vertex $i$ to vertex $j$ if and only if receiver $i$ knows the desired message by receiver $j$. For the most part of this paper we work with this graph model for the index coding problem. 
Observe that in the most general case, one has to work with hypergraphs to represent the side information.

It is common to study the index coding problem in terms of an achievable rate region based on the size of the $m$ messages to be decoded by the $m$ receivers (see Section \ref{sec:definitions} for a formal definition). 
Here the rate of a receiver refers to the normalized amount of information transmitted to it. 
The set of all achievable rates, i.e. the capacity region, for index coding problem remains an open problem. Nonetheless, the problem has been solved in some special cases, notably for the equal-rate case under certain graph structures \cite{HamedViveckSyedAli}. In \cite{KarthikeyanAlexandrosMichael}, the capacity region of an index coding problem is related to some graph theoretical features such as local chromatic number. A difference between the performance of linear and non-linear codes is characterized in \cite{EyalUri}.

\subsection{Connections with Network Coding and Wireless Communication}
The index coding problem has significant connections with network coding and wireless communications. It is clear that every instance of index coding can be represented as an instance of a network coding in which a single node desires to send messages via a unit capacity channel and some channels with infinite capacity representing side information. In \cite{MichelleSalimMichelle} it is shown that for both linear and non-linear case, for any instance of networking coding problem, there exists an instance of index coding problem with the same capacity region. In addition, in \cite{SalimAlexCostas} a reduction from an instance of network coding problem to an instance of index coding problem is introduced. They used this reduction to show that the capacity regions for linear and one-shot cases are not equal to capacity region of asymptotic non-linear case. 

In \cite{SyedAli}, the topological interference management problem is introduced for both wired and wireless networks. In the wireless set up, this problem refers to the analysis of degrees of freedom of an interference network with the assumption that all weak interferences  are zero. This natural problem in the wireless networks has a significant relation to the index coding problem. For example, in \cite{SyedAli} it is proved that the set of degrees of freedom which are available through linear schemes in the topological interference management problem is equal to the linear capacity region of an equivalent index coding problem. Moreover, the non-linear degree of freedom region of the interference management problem is related to the non-linear capacity region of the problem.

\subsection{Our contribution}
Given a fixed set of rates, let $\mathcal{G}$ denote the set of all graphs that support the rates. We are interested in minimal members of $\mathcal{G}$ (with respect to containment of the edge set). More specifically, a graph is said to be \emph{critical} (or \emph{edge critical}) if (1) it belongs to $\mathcal{G}$ and (2) deletion of any edge from the graph makes it to fall outside $\mathcal{G}$. It is useful to study critical graphs since it identifies the minimum-cost architectures of the networks supporting a given set of rates. 

To best of our knowledge, critical graphs for index coding have not been studied before.
We present several results in this paper regarding critical graphs. When the rates are all equal, we identify the critical graph with minimum number of edges (Theorem \ref{thm1}). Next we study
the general case of arbitrary rates via an additivity result that we prove about index coding (Theorem \ref{thm3new}; here we basically prove that a simple time division strategy is optimal).
We use this result to show that critical graphs for one-shot and asymptotic linear index coding as well as those of non-linear asymptotic index coding are structured, by proving that they have to be a union of disjoint strongly connected subgraphs (USCS) (Theorem \ref{thm2}). Equivalently, each directed edge in the graph has to be on a cycle in the graph. On the other hand, for non-linear one-shot index coding, we construct a counterexample by finding a critical graph that is not USCS. In addition, using Theorem \ref{thm3new}, we prove criticality of the union of two critical graphs (Theorem \ref{thm3}). Moreover, we show this result holds for symmetric criticality in both one-shot and asymptotic linear case, as well as in asymptotic non-linear case (Theorem \ref{thm3}). In the next step, we provide a comprehensive list of symmetric critical graphs for graphs with at most five nodes, and use this list identify two general classes of critical graphs which explain many of the critical graphs that we had observed (Theorem \ref{thm5}). Finally, we have generalized some of our results to the groupcast index coding setting (Theorem \ref{thm6}).

A potential application of index coding problem is in the study of wireless broadcast networks. For example, in \cite{NeelySaberZhang} side information of nodes in a broadcast wireless network has been employed to make the communication more efficient. In such schemes, study of critical graphs can be helpful as it identifies the side information that cannot make the communication more efficient. For instance, as our results shows, those side information whose corresponding edge in the side information graph do not lie on any cycle, will not improve the efficiency of communication. Hence, these side information can be eliminated. Accordingly, the total storage resources of wireless nodes can be decreased using our results.

Additionally, even though we are mostly interested in critical graphs in this work, our results address the ``index coding problem'' itself. For instance, our result on the additivity of the capacity region of index coding problem (Theorem \ref{thm3new}) finds the index coding capacity of a graph in terms of those of its subgraphs, if the graph has a certain structure. Further we believe that by studying the characteristics of critical graphs, one can use the capacitiy region of some critical subgraphs of the graph $\mathsf{G}$ to find a lower bound for the index coding problem introduced by graph $\mathsf{G}$.

This paper is organized as follows: in Section \ref{sec:definitions}, we introduce the basic notation and definitions used in this paper. The results are provided in Section \ref{Section:StatementOfResults}. In Subsection \ref{subsec:structures}, some results that suggest structures for critical graphs are given. In addition, In Subsection \ref{subsec:groupcast}, an expansion of the former results for groupcast index coding is presented. Appendix \ref{appendixA} contains a few lemmas used in the proofs, Appendix \ref{appendixB} contains the source file for a C program needed to do an exhaustive search to complete the proof of one of the theorems, and Appendix \ref{appendixC} contains a list of all symmetric rate critical graphs on 5 vertices.


\section{Definitions and Notations}\label{sec:definitions}
A (unicast) index coding problem comprises of $m$ nodes, $\{1, \cdots, m\}$, and a set of $m$ message $\{W_1, \cdots, W_m\}$ where node $i$ needs to decode the message $W_i$, $i=1, \cdots, m$. The side information of node $i$ is assumed to be a subset of $\{W_1, \cdots, W_{i-1}, W_{i+1},\cdots, W_m\}$. We can illustrate this side information by a directed graph $\mathsf{G} = (\mathcal{V}, \mathcal{E})$, where $\mathcal{V} = \{1, \cdots, m\}$ and node $i$ has an edge to node $j$ (that is, $(i, j) \in \mathcal{E}$) if node $i$ knows $W_j$. For simplicity in the rest of this paper, we use graph as a shorthand for directed graphs. Undirected graphs are referred to by the term ``bidirectional graph".

\begin{definition}
A code for an index coding problem (or an index code) consists of
\begin{enumerate}
\item $m$ alphabet sets $\mathcal{W}_i$, $i=1,2,\cdots,m$ where the message intended by the $i$-th party, $W_i$, belongs to $\mathcal{W}_i$;
\item An encoding function $f$ from $\mathcal{W}_1\times\cdots\times\mathcal{W}_m$ to $\{1,2,\cdots, N\}$ that compresses the messages $(W_1, \cdots, W_m)$ into a symbol in $\{1,2,\cdots, N\}$. $f(W_1, \cdots, W_m)$ is called the public message since it will be made available to all the nodes;
\item A set of $m$ decoding functions at the nodes from $\{1,2,\cdots, N\}\times \prod_{(i, j)\in \mathcal{E}}\mathcal{W}_j$ to $\mathcal{W}_i$ for $i=1,2,\cdots, m$. Every node should be able to decode its message using the public message and its side information.
\end{enumerate}
The rate vector associated with the code is a vector $(r_1, \cdots, r_m)$ where
\begin{eqnarray}\label{eqn:newnewn2}
r_i = \frac{\log(|\mathcal{W}_i|)}{\log(N)}.
\end{eqnarray}
We will use $\overline{\bf r}$ to indicate the rate vector $(r_1, \cdots, r_m)$.

Probability of error associated to the code is the probability that node $i$ fails to correctly decode $W_i$ for some $i=1,2,\cdots, m$, where rvs $W_i$ ($i=1,2,\cdots, m$) are assumed to be uniform on their alphabet set and mutually independent of each other.
\end{definition}

Linear codes form a subclass of codes, and are defined as follows:
\begin{definition}
A linear code for an index coding problem with finite field $\mathbb{F}$ consists of
\begin{enumerate}
\item$m$ positive integers $l_1, \cdots , l_m$ indicating that $W_i \in \mathbb{F}^{l_i}$ is a sequence of length $l_i$ of symbols in $\mathbb{F}$. In other words, the alphabet set for the random variable $W_i$ is $\mathcal{W}_i=\mathbb{F}^{l_i}$;
\item A linear map $f$ from $\mathcal{W}_1\times\cdots\times\mathcal{W}_m$ to $\mathbb{F}^n$ that compresses the messages $(W_1, \cdots, W_m)$ into a sequence of length $n$ of symbols in $\mathbb{F}$;
\item A set of $m$ linear decoding functions from $\mathbb{F}^n\times \prod_{(i, j)\in \mathcal{E}}\mathcal{W}_j$ to $\mathcal{W}_i$ for $i=1,2,\cdots, m$.
\end{enumerate}
The rate vector associated with the code is a vector $\overline{\bf r} = (r_1, \cdots, r_m)$ where
\begin{eqnarray}\label{eqn:newnewn1}
r_i = \frac{l_i}{n}.
\end{eqnarray}
\end{definition}

Now, we introduce two classifications for the index coding problem.
\begin{definition}{\bf Linear and Non-Linear Index Coding}\\
In linear index coding we restrict ourselves to linear codes over an arbitrary finite field $\mathbb{F}$. However in the non-linear index coding we are allowed to use an arbitrary code.
\end{definition}
\begin{definition}{\bf One-Shot and Asymptotic Index Coding}\\
In the one-shot problem, we have fixed message alphabets $\mathcal{W}_1, \cdots, \mathcal{W}_m$ and seek the code with the smallest
alphabet size for the public message that can result in a zero probability of error. On the other hand, in the asymptotic coding scheme we are only given rate vector $\overline{\bf r} = (r_1, \cdots, r_m)$. Then there should exist a sequence of codes with zero error probability whose rate vectors converge to $\overline{\bf r} = (r_1, \cdots, r_m)$.

\end{definition}
\begin{remark}
The asymptotic index coding is generally defined for a vanishing probability of error rather than an exactly zero probability of error. However it is shown in \cite{HavivLangberg} that the two definitions are equivalent.\label{rmk:zeroerror}
\end{remark}

\begin{definition}{\bf Critical Graphs and Symmetric Rate Critical Graphs}\\
Given an index coding problem (linear or non-linear/one-shot or asymptotic) on a graph, we say that the graph is \emph{critical} if removal of any edge from it \emph{strictly} shrinks the rate region (capacity when we are looking at asymptotics) associated to the graph. 

The maximum symmetric rate supported by a graph is the supremum of $r$ such that $\overline{\bf r} = (r,r,\cdots, r)$ is achievable. We say that the graph is \emph{symmetric rate critical} if removal of any edge from it \emph{strictly} reduces the maximum symmetric rate by the graph. Every symmetric rate critical graph is critical, but the reverse is not necessarily true (see Theorem \ref {thm:newstr1}).
\end{definition}

Next we need the following definitions from graph theory:
\begin{definition}{\bf Tur\'{a}n Graph}\\
Tur\'{a}n Graph of order $m$ and $k$, denoted by $T(m, k)$, is a bidirectional  complete $k$-partite graph with $b$ parts of size $a + 1$ and $k - b$ parts of size $a$, where $m=ak+b$ for $a\geq 0, b\in \{0,1,2\cdots, k-1\}$. We denote the number of edges of $T(m, k)$ by $e(m, k)$. In \cite[Ex. 5.2.18]{Douglas}, it is shown that 
\begin{eqnarray}
e(m, k) = \frac{1}{2} \cdot (1 - \frac{1}{k})m^2 - \frac{b(k-b)}{2k}.
\end{eqnarray}
See also Lemma \ref{lemma:Turan} from Appendix \ref{appendixA}.
\end{definition}
\begin{definition}{\bf Strongly Connected Graphs}\\
The graph $\mathsf{G} = (\mathcal{V}, \mathcal{E})$ is strongly connected if there exists a directed path between every pair of distinct vertices.
\end{definition}
It is easy to verify that a graph is strongly connected if and only if every edge of the graphs lies on a (directed) cycle.
\begin{example}
The graph shown in Fig. \ref{fig:strongly-connected-graph}{} is strongly connected. However, the graph shown in Fig. \ref{fig:not-connected-graph} is not strongly connected since there is no directed path between nodes 4 and 2. Here the edge from node $4$ to node $5$ does not lie on a directed cycle.
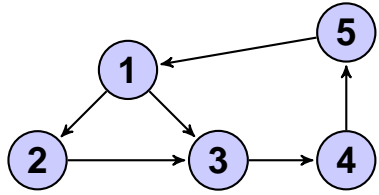
\begin{figure}\centering
\begin{tikzpicture}
[->, >=stealth', shorten >=1pt,auto, node distance=1.7cm, thick,main node/.style={circle,fill=blue!20,draw,font=\sffamily\Large\bfseries}]
\node[main node] (1) {1};
\node[main node] (2) [below left of=1] {2};
\node[main node] (3) [below right of=1] {3};
\node[main node] (4) [right of=3] {4};
\node[main node] (5) [above of=4]{5};
\path[every node/.style={font=\sffamily\small}]
(1) edge node {} (2)
edge node {} (3)
(2) edge node {} (3)
(3) edge node {} (4)
(4) edge node {} (5)
(5) edge node {} (1);
\end{tikzpicture}
\caption{\footnotesize An example of a strongly connected graph.\normalsize } \label{fig:strongly-connected-graph}
\end{figure}
\end{example}
\begin{figure}\centering
\begin{tikzpicture}
[->, >=stealth', shorten >=1pt,auto, node distance=1.7cm, thick, main node/.style={circle,fill=blue!20,draw,font=\sffamily\Large\bfseries}]
\node[main node] (1) {1};
\node[main node] (2) [below left of=1] {2};
\node[main node] (3) [below right of=1] {3};
\node[main node] (4) [right of=3] {4};
\node[main node] (5) [above of=4]{5};
\path[every node/.style={font=\sffamily\small}]
(1) edge node {} (2)
edge node {} (3)
edge node {} (5)
(2) edge node {} (3)
(3) edge node {} (4)
(4) edge node {} (5);
\end{tikzpicture}
\caption{\footnotesize An example of a graph that is not strongly connected.\normalsize } \label{fig:not-connected-graph}
\end{figure}
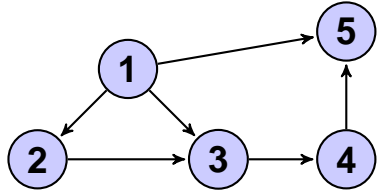
\begin{figure}\centering
\begin{tikzpicture}
[->,>=stealth',shorten >=1pt,auto,node distance=1.7cm, thick,main node/.style={circle,fill=blue!20,draw,font=\sffamily\Large\bfseries}]
\node[main node] (1) {1};
\node[main node] (2) [below left of=1] {2};
\node[main node] (3) [below right of=1] {3};
\node[main node] (4) [right of=3] {4};
\node[main node] (5) [above of=4]{5};
\node[main node] (6) [right of=4]{6};
\path[every node/.style={font=\sffamily\small}]
(1) edge node {} (2)
edge [bend left] node {} (3)
(2) edge node {} (3)
(3) edge [bend left] node {} (1)
(4) edge [bend left] node {} (5)
(5) edge [bend left] node {} (4);
\end{tikzpicture}
\caption{\footnotesize A USCS graph\normalsize } \label{fig:uscs-graph}
\end{figure}
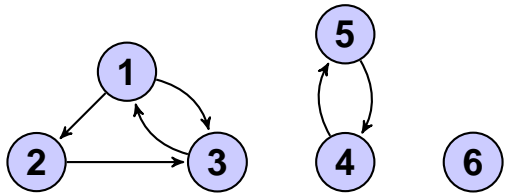
\begin{definition}{\bf Union of Two Disjoint Graphs}\\
The union of $\mathsf{G} = (\mathcal{V}, \mathcal{E})$ and $\mathsf{G}'= (\mathcal{V}', \mathcal{E}')$ is defined as $\mathsf{G} \cup \mathsf{G}' = (\mathcal{V} \cup \mathcal{V}', \mathcal{E} \cup \mathcal{E}')$.
\end{definition}
\begin{definition}{\bf USCS Graphs}\\
Graph $\mathsf{G}$ is {\sl USCS} (Union of Strongly Connected Subgraphs) if there exists a set of disjoint graphs $\{\mathsf{G}_1, \mathsf{G}_2, \cdots , \mathsf{G}_k\}$ such that (1) $\mathsf{G}_i$ is strongly connected and (2)
$
\mathsf{G} = \underset{i}{\bigcup}\mathsf{G}_i
$.
\end{definition}
\begin{example}
Because the graph shown in Fig. \ref{fig:strongly-connected-graph} is strongly connected, it is USCS, too. However, the graph shown in Fig. \ref{fig:not-connected-graph} is not USCS. Next consider the graph shown in Fig. \ref{fig:uscs-graph}. If we define $\mathsf{G}_1$ as the induced subgraph of the set $\{1, 2, 3\}$, $\mathsf{G}_2$ as the induced subgraph of the set$\{4, 5\}$, and $\mathsf{G}_3$ as the induced subgraph of the set $\{6\}$, then we have $\mathsf{G} = \mathsf{G}_1 \cup \mathsf{G}_2 \cup \mathsf{G}_3$. Therefore, due to the fact that $\mathsf{G}_1$, $\mathsf{G}_2$, and $\mathsf{G}_3$ are strongly connected, $\mathsf{G}$ is USCS.
\end{example}


\section{Main Results}\label{Section:StatementOfResults}

\begin{theorem}{\bf Minimum Number of Edges for Equal Rates}\\ \label{thm1}
Every $m$-vertex graph supporting a rate vector $\overline{\bf r} = (r, \cdots, r)$ has at least:
\begin{eqnarray}
g(r, m) = m(m-1) - 2\cdot e(m, \left \lfloor \frac{1}{r} \right \rfloor)
\end{eqnarray}
edges, if $\frac{1}{m}\leq r \leq 1$ ($g(r, m)$ is the number of edges in the complement of $T(m, \left \lfloor \frac{1}{r} \right \rfloor)$). Moreover, there is a unique graph, up to isomorphism, that has exactly $g(r, m)$ edges and supports the rate vector $\overline{\bf r} = (r, \cdots, r)$. This theorem holds for all cases (linear or non-linear, one-shot or asymptotic).
\end{theorem}
\begin{remark}This theorem shows that there is a unique (up to isomorphism) critical graph with minimum number of edges for both one-shot and asymptotic cases.
\end{remark}
\begin{remark} \label{rmk3}
Theorem \ref{thm1} is valid for $\frac{1}{m}\leq r \leq 1$. For the case $r > 1$, there is no graph that supports the rate vector $\overline{\bf r} = (r, \cdots, r)$ since the rate of each node cannot be greater than one. When $r < \frac{1}{m}$,  it is possible to send all messages as the public message, and as a result, there is no need to have any side information. Therefore, the empty graph is sufficient in this case.
\end{remark}

\begin{theorem}{\bf Additivity of index coding capacity region}\label{thm3new}

\textbf{a)} Given a graph $\mathsf{G} = (\mathcal{V}, \mathcal{E})$, suppose that $\mathsf{G'}$ and $\mathsf{G''}$ are subgraphs of $\mathsf{G}$ induced on vertex sets $\mathcal{V'}$ and $\mathcal{V''}$. In addition, assume that $\mathcal{V'}$ and $\mathcal{V''}$ partition $\mathcal{V}$ and there exist no edge like $e = (u, v)$ in $\mathcal{E}$ that starts from $u \in \mathcal{V'}$ and ends up in $v \in \mathcal{V''}$, i.e. no directed edge from  $\mathsf{G}'$ to $\mathsf{G}''$ exists. Then, elimination of all the directed edges from $\mathsf{G}''$ to $\mathsf{G}'$
 will not change the rate region in the one-shot linear, and in the asymptotic non-linear index coding problems.

\textbf{b)} [Optimality of a simple time-division strategy]. Take an index coding problem with graph $\mathsf{G}=\mathsf{G}'\bigcup \mathsf{G}''$, such that there is no edge between $\mathsf{G}'$ and $\mathsf{G}''$.
Let $\mathcal{C}$, $\mathcal{C}'$ and $\mathcal{C}''$ denote the capacity regions of $\mathsf{G}$, $\mathsf{G}'$ and $\mathsf{G}''$ respectively (the three capacities are all either in the sense of asymptotic linear, or all in the sense of asymptotic non-linear). Then $\mathcal{C}=\bigcup_{\alpha\in[0,1]} \alpha \mathcal{C}'\oplus (1-\alpha)\mathcal{C}''$ where $\oplus$ is the direct sum operator. Alternatively, the index coding region for $\mathsf{G}$ is of the form $\overline{\bf r}=(\alpha \overline{\bf r}', (1-\alpha)\overline{\bf r}'')$ for $\alpha\in[0,1]$ and vector $\overline{\bf r}'$ is in the region of $\mathsf{G}'$ and $\overline{\bf r}''$ is in the region of $\mathsf{G}''$, and $(\alpha \overline{\bf r}', (1-\alpha)\overline{\bf r}'')$ is the concatenation of the vectors $\alpha \overline{\bf r}'$ and $(1-\alpha)\overline{\bf r}''$.
\end{theorem}

\begin{theorem}{\bf Critical graphs are USCS}\label{thm2} 

\textbf{a)} Every critical graph for linear index coding (one-shot or asymptotic)
and for asymptotic non-linear index coding is USCS. In particular, removing edges that do not lie on a directed cycle does not change the capacity region in these cases.

\textbf{b)} There exists a critical graph for a one-shot non-linear index coding problem which is not USCS.
\end{theorem}

The condition given in item (a) of Theorem \ref{thm2} are necessary but not necessarily sufficient, i.e. USCS does not necessarily imply criticality. This follows from the fact that if we add an edge to a USCS graph that supports a given set of rates, the resulting graph remains a USCS graph that still supports the given rates. However observe that the resulting graph, with one more additional edge, may indeed support higher rates. This observation may lead one to propose the following modified conjecture:
\begin{conjecture}Take a USCS graph that supports a given set of rates for asymptotic non-linear index coding. Let $e$ be an edge of the graph that lies on a single directed cycle (i.e. it is completing a cycle and its removal breaks that cycle). Then removing the edge $e$ from the graph results in a graph that no longer supports the given set of rates.
\end{conjecture}

However the above conjecture is also false. Consider the graph in Fig. \ref{fig:a22} with $\mathcal{S}=\{0,1\}$.
\begin{figure}\centering
\begin{tikzpicture}
[->,>=stealth',shorten >=1pt,auto,node distance=2.5cm, thick,main node/.style={circle,fill=blue!20,draw,font=\sffamily\Large\bfseries}]
\node[main node] (1) {1};
\node[main node] (2) [below left of=1] {2};
\node[main node] (3) [below right of=1] {3};
\path[every node/.style={font=\sffamily\small}]
(1) edge node {} (2)
edge [bend left] node {} (3)
(2) edge node {} (3)
edge [bend left] node {} (1)
(3) edge node {} (1);
\end{tikzpicture}
\caption{\footnotesize A counterexample to Conjecture 1. \normalsize} \label{fig:a22}
\end{figure}
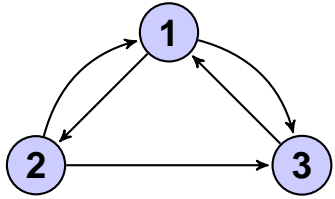 
Using Lemma \ref{lemma:rates-sum} of the Appendix \ref{appendixA}, the sum of the rate of every subset of nodes which contains no cycle should be less or equal to one. Therefore, for every rate vector $\overline{\bf r} = (r_1, r_2, r_3)$ supported by this graph, we have:
\begin{eqnarray}
r_2 + r_3 \leq 1 ,\label{eq:converse-example-equality-1} \\
r_1 \leq 1. \label{eq:converse-example-equality-2}
\end{eqnarray}
The edge from node $2$ to node $3$ lies on a unique cycle $2\rightarrow 3\rightarrow 1\rightarrow 2$. We show that if this edge is removed, all rate vectors satisfying eqs. \eqref{eq:converse-example-equality-1} and \eqref{eq:converse-example-equality-2} will be still supported. It suffices to prove that any $\overline{\bf r} = (r_1, r_2, r_3)$ satisfying $r_1=1$ and $r_2+r_3=1$ is supported by the new graph. If we assume that $W_i$ is a binary string of length $l_i$ where $l_1=l_2+l_3$, we can create $f(W_1, W_2, W_3)$ as follows: we concatenate $W_2$ and $W_3$ to create a binary string of length $l_2+l_3=l_1$ and then XOR it with the binary string of $W_1$. Node $1$ knows both $W_2$ and $W_3$ and hence can recover $W_1$. And both nodes $2$ and $3$ know $W_1$. Hence they can both recover their desired message.

\begin{remark}
There has been some previous work on the effect of edge removal in network coding \cite{HavivLangberg}-\cite{ShirinMichelleTracey}.
However to best of our knowledge there is no previous work on edge removal in the context of index coding.
\end{remark}

\subsection{Structure of Critical Graphs}\label{subsec:structures}

In this section we provide some results on the structure of critical graphs. The first class of critical graphs that are easy to identify are bidirectional graphs:

\begin{theorem} \label{thm:newstr1} Any bidirectional graph is critical (by a bidirectional graph we mean one in which a directed edge from node $i$ to $j$ implies a directed edge from node $j$ to $i$). On the other hand this is not true of symmetric criticality; in particular a bidirectional cycle of size 4 is not symmetric critical.
\end{theorem}

To derive the main results for this section, we first produced all symmetric rate critical graphs for graphs on 5 vertices. This list was compiled using the data available on Young-Han Kim's personal website \cite{YHKpersonalwebsite}, and is given in  Appendix \ref{appendixC}. We then tried to formulate a few theorems that would explain the structure of critical graphs that we observed. 

\begin{theorem}{\bf Union of two critical graphs is critical}\\ \label{thm3}
If $\mathsf{G}$ and $\mathsf{H}$ are two critical graphs, then $\mathsf{G} \cup \mathsf{H}$ is also a critical graph for any of linear/non-linear, one-shot/asymptotic formulations. Further,
if $\mathsf{G}$ and $\mathsf{H}$ are two symmetric rate critical graphs, then $\mathsf{G} \cup \mathsf{H}$ is also a symmetric rate critical in one-shot linear, asymptotic linear, and asymptotic non-linear index coding scenarios. \end{theorem}

\begin{theorem} {\bf Two structures that are critical}\\ 
\label{thm5}
\textbf{a)} Suppose $\mathsf{G} = (\mathcal{V}, \mathcal{E})$ is a directed cycle of length $n$, where
$$ \mathcal{V} = \{1, \cdots, m\},$$
$$\mathcal{E} = \{(i, i+1): 1\leq i < m\} \cup \{(m, 1)\}.$$
Now, construct a new graph $\mathsf{G}' = (\mathcal{V'}, \mathcal{E'})$ so that $\mathcal{V'} = \mathcal{V} \cup \{m+1\}$ and $\mathcal{E'} = \mathcal{E} \cup \{(m+1, 1), (m+1, i), (j, m+1), (k, m+1)\}$. Then, if $1 \leq j  < i$ and $i \leq k \leq n$, $\mathsf{G}'$ is symmetric rate critical.

\textbf{b)} Suppose $\mathsf{G}'=(\mathcal{V'}, \mathcal{E'})$ is a graph that satisfies the condition of part (a). We construct a new graph $\mathsf{G}''=(\mathcal{V''}, \mathcal{E''})$ by replacing any vertex $u \in \mathcal{V'}$ by a complete graph 
(different vertices can be replaced by complete graphs of different sizes). Then, $\mathsf{G}''$ is critical. More specifically, we replace vertex $u$ with $n_u$ vertices $(u,1), (u,2), \cdots, (u,n_u)$ that are mutually connected to each other. We also draw a directed edge from $(u, i)$ to $(v, j)$ in $\mathsf{G}''$ for $i\in[1:n_u]$ and $j\in[1:n_v]$ if there exists a directed edge from $u$ to $v$ in $\mathsf{G}'$.
\end{theorem}

\begin{remark}\label{rmk5a}
The criticallity of graphs Fig.\ref{fig:893}, Fig.\ref{fig:771}, Fig.\ref{fig:2312}, Fig.\ref{fig:2120}, Fig.\ref{fig:2404}, Fig.\ref{fig:8625}, Fig.\ref{fig:8495}, Fig.\ref{fig:8285}, Fig.\ref{fig:7026}, and Fig.\ref{fig:8847} can be shown by Theorem \ref{thm5}. 
\end{remark}

\subsection{Extension to Groupcast}\label{subsec:groupcast}

The index coding problem that we considered so far is called unicast index coding problem. A generalization of the unicast index coding is the groupcast index coding. In groupcast index coding, the desired messages of receivers are not necessarily disjoint, i.e. a group of receivers can desire the same message.

\begin{definition}{\bf Groupcast Index Coding}\\
Assuming a set of $m$ messages $\{W_1, W_2, \cdots, W_m\}$, a groupcast index coding problem can be modeled with a directed hypergraph on $m$ vertices with node $i$ representing $W_i$. Each receiver can be represented as a directed hyperedge starting from its desired message and ending at its side information. In other words, if receiver $i$ wants to know $W_{d_i}$ while having $A_i\subset \{W_1, W_2, \cdots, W_m\} \backslash W_{d_i}$ like $A_i$ as its side information, we add a directed hyperedge from $\{W_{d_i}\}$ to $A_i$. The number of receivers will be equal to the number of hyperedges.

A hypergraph is said to be critical if eliminating any member of the side information set of any receiver strictly reduces the set of rates supported by the hypergraph.

\end{definition}

\begin{definition}{\bf Underlying Digraph of a Directed Hypergraph}\\
Let $\mathsf{H}=(\mathcal{V}, \mathcal{E})$ be a directed hypergraph. Then we call $\mathsf{G}=(\mathcal{V}, \mathcal{E}_\mathsf{G})$ the underlying digraph (directed graph) of $\mathsf{H}$, where:
$$\mathcal{E}_\mathsf{G} = \{(u, v) \  | \  \exists P,Q \subseteq \mathcal{V}: ~~  u\in P, v\in Q,~~  (P, Q)\in \mathcal{E}\}$$
\end{definition}
\begin{remark}
Since groupcast index coding problem is a generalization of the unicast index coding problem, we can define the side information hypergraph for the unicast index coding problem too. It can be easily verified that the underlying digraph of this hypergraph is equal to the directed graph we used to model unicast index coding problem.
\label{rmk:gic-hypergraph}
\end{remark}

\begin{theorem}{\bf Groupcast Critical Graphs are USCS too} \label{thm6}

\textbf{a)} The underlying graph of every critical hypergraph for linear groupcast index coding (one-shot or asymptotic), and for asymptotic non-linear groupcast is USCS.

\textbf{b)} There exists a critical hypergraph for a one-shot non-linear groupcast index coding problem which is not USCS.
\end{theorem}

\section{Future Work}\label{Section:future}
Consider the index coding for a random graph where  directed edges exists between any two nodes with probability $p$ and independent of other edges. Computing index coding for this class of random graphs can be of interest. Theorem \ref{thm2}  can be used to find a lower bound on the expected number of edges that we can remove from this graph such that it does not affect the capacity region. A lower bound is the expected value of number of edges that do not lie on a directed cycle, which is equal to $2{n\choose 2}$ times the probability that a directed edge from node $1$ to node $2$ exists which does not lie on a directed cycle. The expected value will be equal to $2{n\choose 2}p(1-q)$ where $q$ is the probability that there is a directed path from node $1$ to node $2$; we have multiplied $p$ with $1-q$ as they correspond to independent events. Computing $q$, pair connectedness in directed random graphs, is a studied topic in percolation theory \cite{percolation-paper} but we were not able to find a closed form formula for it.


\section{Proofs}
\subsection {Proof of Theorem \ref{thm1}}
We begin by proving the given lower bound on the minimum number of edges. It suffices to prove it for the non-linear asymptotic case since it implies that for all other cases. Suppose that a given graph $\mathsf{G} = (\mathcal{V}, \mathcal{E})$ supports the rate vector $\overline{\bf r} = (r, \cdots, r)$ for non-linear asymptotic case. We aim to construct two new graphs and with the help of Lemma \ref{lemma:rates-sum} and \ref{lemma:Turan} find some bounds on the number of edges in these two graphs. Then we use these bounds to find a bound on the number of edges in $\mathsf{G}$. Using Lemma \ref{lemma:rates-sum}, every subset of $\mathcal{V}(\mathsf{G})$ whose size is bigger than $\left \lfloor \frac{1}{r}\right \rfloor$, has a directed cycle, because the sum of the rates of the vertices in this subset is greater than or equal to $r\times(\left \lfloor\frac{1}{r} \right \rfloor + 1) > 1$. Then, we consider an arbitrary order for the vertices of $\mathsf{G}$ such as ${1}, \cdots, {m}$ and construct two new graphs (called ``forward" and ``backward" graphs) as follows: $\mathsf{G}^f = (\mathcal{V}^f, \mathcal{E}^f)$ and $\mathsf{G}^b = (\mathcal{V}^b, \mathcal{E}^b)$ where $\mathcal{V}^f = \mathcal{V}^b = \mathcal{V}$, and $\mathcal{E}^f,\mathcal{E}^ b$ is a partition of $\mathcal{E}$ into two sets as follows: $\mathsf{G}^f$ contains those edges of $\mathsf{G}$ whose direction agrees with the mentioned order, that is, $\mathcal{E}^f = \{(x, y) \in \mathcal{E}| x<y\}$. $\mathsf{G}^b$ contains the following edges: $\mathcal{E}^b = \{(x, y) \in \mathcal{E}| x>y\}$. Now, because every cycle in $\mathsf{G}$ should contain at least one edge from both $\mathsf{G}^f$ and $\mathsf{G}^b$, every subset of size more than $\left \lfloor \frac{1}{r}\right \rfloor$ has at least one edge in both $\mathsf{G}^f$ and $\mathsf{G}^b$.

Now let us construct a bidirectional graph $\widetilde{\mathsf{G}}^f$ on the same set of vertices as follows: $x$ is connected to $y$ in $\widetilde{\mathsf{G}}^f$ for $x\neq y$ if an only if $({\min(x,y)}, {\max(x,y)})\notin \mathcal{E}^f$. Observe that $\widetilde{\mathsf{G}}^f$ is like the complement of $\mathsf{G}^f$ if we ignore the edge arrows of $\mathsf{G}^f$. Similarly, $\widetilde{\mathsf{G}}^b$ is constructed as the complement of $\mathsf{G}^b$ if we ignore the direction of arrows in it. Since every subset of size more than $\left \lfloor \frac{1}{r}\right \rfloor$ has at least one edge in both $\mathsf{G}^f$ and $\mathsf{G}^b$, we can conclude that $\widetilde{\mathsf{G}}^f$ and $\widetilde{\mathsf{G}}^b$ do not have a clique of size $\left \lfloor \frac{1}{r}\right \rfloor+1$. Using Lemma \ref{lemma:Turan}, the number of edges of both $\mathsf{G}^f$ and $\mathsf{G}^b$ is at least
\begin{eqnarray}
\binom{m}{2} - e(m, \left \lfloor\frac{1}{r}\right \rfloor)=\frac{g(r, m)}{2}.
\end{eqnarray}
Hence, $\mathsf{G}$ itself has at least $g(r, m)$ edges.\\

Next, we will show that the complement of $T(m, \left \lfloor\frac{1}{r}\right \rfloor)$ supports the rate $\overline{\bf r}$. It suffices to show this for one-shot linear coding and it implies that for all cases there exists a graph which supports the rate $\overline{\bf r}$. 
Let $m=a\left \lfloor \frac{1}{r}\right \rfloor+b$ for some $a\geq 0, b\in \{0,1,2\cdots, \left \lfloor \frac{1}{r}\right \rfloor-1\}$.
 Then  we construct $\mathsf{G}$ so that it consists of $b$ cliques of size $a+1$, and $\left \lfloor \frac{1}{r}\right \rfloor - b$ cliques of size $a$.\footnote{A clique is a graph where every vertex has a directed edge to every other vertex.} Then one can verify that $\mathsf{G}$ has $g(r,m)$ edges. In addition, if every node desires only one bit and we transmit the XOR of the bits in every clique, every vertex can decode its message, and the rate of every message equals to $\frac{1}{\lfloor \frac{1}{r}\rfloor} \geq r$. Furthermore, it is obvious that this is a one-shot linear coding. Thus we have shown that there exists a graph which supports the rate $\overline{\bf r}$. \\

Lastly, to show that no other graph with exactly $g(r, m)$ edges supports $\overline{\bf r}$,  consider a graph $\mathsf{G}$ that has $g(r, m)$ edges and supports the rate vector $\overline{\bf r} = (r, \cdots, r)$ in non-linear asymptotic case (it suffices to show this for the non-linear asymptotic case and it will imply other cases). If we construct $\mathsf{G}^f$ and $\mathsf{G}^b$ as discussed before, each of them should have exactly $\frac{g(r,m)}{2}$ edges and they should have the structure mentioned in Lemma \ref{lemma:Turan}. So, the only remaining step is to show that the cliques in $\mathsf{G}^f$ and $\mathsf{G}^b$ coincide on each other. Suppose this does not hold, that is, there are two vertices where there is an edge between them in $\mathsf{G}^f$, but not in $\mathsf{G}^b$. Let us call these two vertices $u$ and $v$. Choose one vertex from each of the $\lfloor \frac{1}{r} \rfloor$ components of $\mathsf{G}^f$ such that $u$ is chosen and let us denote this set by $X$. Then we claim that $X \cup \{v\}$ does not contain any cycle in $\mathsf{G}$. Note that if a cycle exists, it should include the edge between $u$ and $v$, because it is the only edge in $X \cup \{v\}$ in $\mathsf{G}^f$ and the cycle should have at least one edge from $\mathsf{G}^f$. Now the other edges in the cycle form a path from $v$ to $u$ in $\mathsf{G}^b$ . As every component of $\mathsf{G}^b$ is a clique then $u$ and $v$ should have an edge, which contradicts our assumption that $u$ and $v$ are disconnected in $\mathsf{G}^b$. \mbox{}\hspace*{\fill}\nolinebreak\mbox{$\rule{0.6em}{0.6em}$}

\subsection{Proof of Theorem \ref{thm3new}} 
\subsubsection{Proof of part (a)}\-

\begin{proof}[Proof of part (a) for asymptotic non-linear index coding]
Consider an arbitrary code on the original graph with zero probability of error. Let $K=f(W_1, W_2, \cdots, W_m)$ be the public message. The rate of this code is $\overline{\bf r} = (r_1, r_2, \cdots, r_m)$ where
$$r_i=\frac{ \log(|\mathcal{W}_i|)}{\log(|\mathcal{K}|)}.$$

The union of $\mathsf{G'}$ and $\mathsf{G''}$ corresponds to the graph $\mathsf{G}$ after elimination of  directed edges from $\mathsf{G}''$ to $\mathsf{G}'$. Take an arbitrary $\epsilon>0$. We create a code for the union of $\mathsf{G'}$ and $\mathsf{G''}$ that achieves the rate vector $\overline{\bf r}' = (r'_1, r'_2, \cdots, r'_m)$ where $r'_i\geq r_i-\epsilon$, with the probability of error being less than $\epsilon$. This concludes the proof (see Remark \ref{rmk:zeroerror} on index coding with a vanishing probability of error).

We can conceive $n$ i.i.d. repetitions of the given code with $(W_1^n, W_2^n, \cdots, W_m^n)$ and public message $K^n$. The rate of the i.i.d. code is the same as the original one since
$$\log(|\mathcal{W}_i^n|)=n \log(|\mathcal{W}_i|),\qquad \log(|\mathcal{K}^n|)=n\log(|\mathcal{K}|).$$
Since the original code had zero error probability, the i.i.d. code has also a zero probability of error.

We define $W_\mathsf{G'}$ as a shorthand for $W_i, i\in \mathsf{G'}$, and $W^n_\mathsf{G'}$ as a shorthand for $W_i^n, i\in \mathsf{G'}$. We define a new code that uses $(K', K'')$ instead of $K^n$ where $K'$ is used by nodes in $\mathsf{G'}$ and $K''$ is used by nodes in $\mathsf{G''}$:
\begin{itemize}
\item Size of the alphabet of $K'$, i.e. $|\mathcal{K}'|$, is less than or equal to $2^{n(I(K;W_\mathsf{G'})+\delta)}$. Furthermore, the nodes in $\mathsf{G'}$ can use $K'$ and their side information (which is inside $\mathsf{G'}$) to recover their message with probability $1-\epsilon$.
\item Size of the alphabet of $K''$, i.e. $|\mathcal{K}''|$, is less than or equal to $2^{n(H(K|W_\mathsf{G'})+\delta)}$. Furthermore, the nodes in $\mathsf{G''}$ can use $K''$ and part of their side information of messages inside $\mathsf{G''}$ to recover their message with probability $1-\epsilon$.
\end{itemize}
This would finish the proof since $\log(|\mathcal{K}'|\cdot |\mathcal{K}''|)$ is equal to $n(\log(|\mathcal{K}|)+2\delta)$ and by choosing $\delta$ small enough we can ensure that the rate of the new code is within $\epsilon$ of the original code.

\emph{Construction of $K''$:}

We have $\min_{w_\mathsf{G'}}H(K|W_\mathsf{G'}=w_\mathsf{G'})\leq H(K|W_\mathsf{G'})$. Thus, it suffices to construct $K''$ whose alphabet size is less than or equal to $2^{n(H(K|W_\mathsf{G'}=w_\mathsf{G'})+\delta)}$ where $w_\mathsf{G'}$ is the one that minimizes $H(K|W_\mathsf{G'}=w_\mathsf{G'})$.

Let us first assume in the original problem that $W_\mathsf{G'}=w_\mathsf{G'}$ has occurred and the nodes in $\mathsf{G''}=\mathsf{G}-\mathsf{G'}$ are all aware of this (thus, if some of the nodes in $\mathsf{G''}$ had partial information about messages of nodes in $\mathsf{G'}$, we are giving all of them a full access to $W_\mathsf{G'}$ and this should only help them in decoding their message). Thus the nodes in $\mathsf{G''}$ should be able to recover their intended messages using $K$ and their side information inside $\mathsf{G''}$ with probability one, when $W_\mathsf{G'}=w_\mathsf{G'}$ is fixed. We can use the conditional joint pmf $p(K,W_{\mathsf{G''}}|W_\mathsf{G'}=w_\mathsf{G'})$ as a joint pmf on $q(k,w_{\mathsf{G''}})$ on $K, W_{\mathsf{G''}}$ and think of it as 
 an index code 
on nodes in $\mathsf{G''}$ (since $W_{\mathsf{G''}}$ is independent of $W_\mathsf{G'}$, the marginal distribution of $q(w_{\mathsf{G''}})$ is uniform and coordinatewise mutually independent). The public message in the index coding problem on $\mathsf{G''}$ would be produced according to $q(k)=p(k|W_\mathsf{G'}=w_\mathsf{G'})$ and it leads to zero error probability.

If we have $n$ i.i.d. copies of the pmf $q$ (still a code with zero error probability), the corresponding public message can be compressed using Shannon's source coding theorem and sent to the parties, where nodes in $\mathsf{G''}$ can first decompress it and then use it to run their decoding algorithm. Compression can be achieved at a rate of $H_q(K)+\delta=H(K|W_H=w_H)+\delta$ bits at the cost of a probability of error of $\epsilon$, which is tolerated.

Note that the public message $K''$ is only meant for the use of subgraph $\mathsf{G''}$; to construct the code for $\mathsf{G''}$ we have pretended that $W_\mathsf{G'}=w_\mathsf{G'}$ has happened in each copy of $\mathsf{G'}$. It is clear that $K''$ contains no useful information about $W_\mathsf{G'}^n$ that has actually occurred, and nodes in $\mathsf{G'}$ can ignore $K''$.

\emph{Construction of $K'$:}

Let $p(k, w_\mathsf{G'})$ denote the joint distribution of $K$ and $W_\mathsf{G'}$ in the original code. The decoding function used by node $i\in \mathsf{G'}$ can be expressed as the conditional pmf $p(\hat{w}_i|k, (w_j)_{j:(i,j)\in\mathcal{E}})$ where $\hat{W}_i$ is the reconstruction of node $i$. Of course $\hat{W}_i=W_i$ since perfect reconstruction is assumed. Therefore the joint pmf
$$p(k, w_\mathsf{G'}, \hat{w}_\mathsf{G'})=p(k,w_\mathsf{G'})\prod_{i\in \mathsf{G'}}p(\hat{w}_i|k, (w_j)_{j:(i,j)\in\mathcal{E}})$$
has the property that the marginal distribution on $W_\mathsf{G'}$ and $\hat{W}_\mathsf{G'}$ is equal to \begin{align}p(W_\mathsf{G'}=w_\mathsf{G'}, \hat{W}_\mathsf{G'}=\hat{w}_\mathsf{G'})=\prod_{i\in \mathsf{G'}}\textbf{1}[w_i=\hat{w}_i].\label{eqn:amsm1}\end{align}

We use the covering lemma (rate-distortion coding) to create a code for nodes in $\mathsf{G'}$. Let $\delta>0$ be an arbitrary small positive real.

 \emph{Codebook generation:} Assume that the transmitter and the receivers initially 
share a codebook of $2^{n(I(K;W_\mathsf{G'})+\delta)}$ sequences $$K^n(1), K^n(2), \cdots, K^n(2^{n(I(K;W_\mathsf{G'})+\delta)})$$ each being an i.i.d. sequence according to $p(k)$. 

\emph{Encoding:} Having $W_\mathsf{G'}^n$ at the transmitter, it finds an index $j$ such that $K^n(j)$ is jointly typical with $W_\mathsf{G'}^n$ (i.e. $(K^n(j), W_\mathsf{G'}^n)\in\mathcal{T}_{\delta}^n(p(k,w_\mathsf{G'}))$), where we use the notion of typicality given in \cite[2.4]{AbbasYoungHan}. Since the number of generated $K^n(\cdot)$ sequences is larger than $2^{n(I(K;W_\mathsf{G'})+\delta)}$ by the covering lemma \cite[Lemma 3.3]{AbbasYoungHan}, this can be done with high probability. The transmitter then sends the index $j$ as $K'$ to the receiver (the cardinality of the alphabet of $K'$ allows it to send the index $j$). 

\emph{Decoding:} Having received $K'=j$, nodes $i\in \mathsf{G'}$ create $\hat{W}_i^n$ as a function of $K^n(j)$ 
and their side information (they use the same decoding functions of the original code). More precisely, if we denote the joint pmf of $K^n(j)$ and $W_\mathsf{G'}^n$ by $q_{K^n(j), W_\mathsf{G'}^n}(k, w_\mathsf{G'}^n)$, the joint pmf of the constructed rv's is equal to
$$q_{K^n(j), W_\mathsf{G'}^n}(k, w_\mathsf{G'}^n)\prod_{i\in \mathsf{G'}}\prod_{s=1}^np(\hat{w}_{is}|k_s, (w_{js})_{j:(i,j)\in\mathcal{E}})$$
If $(K^n(j), W_\mathsf{G'}^n)\in\mathcal{T}_{\delta}^n(p(k,w_\mathsf{G'}))$, with high probability we will have $(K^n(j), W_\mathsf{G'}^n,\hat{W}_\mathsf{G'}^n)\in\mathcal{T}_{\delta'}^n(p(k,w_\mathsf{G'},\hat{w}_\mathsf{G'}))$ for any $\delta'>\delta$, 
as we have passed $K^n(j), W_\mathsf{G'}^n$ through the i.i.d. conditional pmf of $p(\hat{w}_\mathsf{G'}|k,w_\mathsf{G'})$ (Conditoinal typicality lemma \cite[2.5]{AbbasYoungHan}).
 Therefore 
 $K^n(j), W_\mathsf{G'}^n, \hat{W}_\mathsf{G'}^n$ will be joint typical with high probability. Thus for any $i\in \mathsf{G'}$, with high probability $(W_i^n, \hat{W}_i^n)$ will be jointly typical. We claim that two sequences $(W_i^n, \hat{W}_i^n)$ jointly typicality in the sense of \cite[2.4]{AbbasYoungHan} is equivalent with their equality. Equation \eqref{eqn:amsm1} implies that 
$p(W_\mathsf{G'}=w_\mathsf{G'}, \hat{W}_\mathsf{G'}=\hat{w}_\mathsf{G'})>0$ if and only if $w_\mathsf{G'}=\hat{w}_\mathsf{G'}$, and hence for any pair  $(w_\mathsf{G'},\hat{w}_\mathsf{G'})$ where $w_\mathsf{G'}\neq \hat{w}_\mathsf{G'}$ we have (using notation of \cite{AbbasYoungHan}) that$$\big|\Pi(w_\mathsf{G'},\hat{w}_\mathsf{G'}|W_i^n, \hat{W}_i^n)-p(w_\mathsf{G'}, \hat{w}_\mathsf{G'})\big|\leq \delta' \cdot p(w_\mathsf{G'}, \hat{w}_\mathsf{G'})=0.$$
Hence $\Pi(w_\mathsf{G'},\hat{w}_\mathsf{G'}|W_i^n, \hat{W}_i^n)=p(w_\mathsf{G'}, \hat{w}_\mathsf{G'})=0$ for any $w_\mathsf{G'}\neq \hat{w}_\mathsf{G'}$, implying that $W_i^n=\hat{W}_i^n$. Therefore with high probability the decoders will successfully decode their intended messages.
\end{proof}

It should be noted that the above proof does not work for asymptotic linear index coding because $K'$ will not necessarily be linear if $K$ is a linear index code.

\begin{proof}[Proof of part (a) for one-shot linear index coding]
Assume that there exists a valid one-shot linear coding scheme for a graph $\mathsf{G}$ with $|\mathcal{V}| = m$ vertices such that:
\begin{eqnarray}
W_i = (w_{i1}, w_{i2}, \cdots , w_{il_i}),
\end{eqnarray}
where $w_{ij} \in \mathbb{F}$ for some field $\mathbb{F}$. Additionally, assume that:
\begin{equation}
f(W_1, W_2, \cdots , W_m) = (t_1, t_2, \cdots , t_n)
\end{equation}
where $t_k$ is equal to
\begin{equation}
t_k = \sum_{i=1}^{m} \sum_{j=1}^{l_i} c_{ijk} \cdot w_{ij},\qquad \forall 1\leq k \leq n,
\end{equation}
for some coefficients $ c_{ijk}$ in the field $\mathbb{F}$. In other words, the following matrix is used for the linear map:
$$\mathsf{C}=
\left[ \begin{matrix}
c_{111} & c_{121} & \cdots & c_{1l_11} & c_{211} & \cdots & c_{ml_m1}\\
&&&\vdots&&&\\
c_{11n} & c_{12n} & \cdots & c_{1l_1n} & c_{21n}  & \cdots & c_{ml_mn}
\end{matrix} \right].
$$

Without loss of generality we can assume that $\mathsf{C}$ is in the row echelon form, since elementary row operation on $\mathsf{C}$ is equivalent to using \emph{invertible} linear combinations of $t_1, t_2, \cdots , t_n$ instead of these variables. The row echelon form can be represented by a sequence of indices \begin{align}(\mathtt{i}_k, \mathtt{j}_k),\qquad k=1,2,\cdots, n, \mathtt{j}_k\leq l_{\mathtt{i}_k}\label{eqn:aaaa1}\end{align}
that are increasing in a lexicographical order, i.e. either $\mathtt{i}_k<\mathtt{i}_{k+1}$ holds or both $\mathtt{i}_k=\mathtt{i}_{k+1}$ and $\mathtt{j}_k< \mathtt{j}_{k+1}$ hold. Further we must have $c_{ijk}=0$ if $(i,j)$ is less than $(\mathtt{i}_k, \mathtt{j}_k)$ in the lexical order.

Since all nodes are able to decode their messages via $(t_1, \cdots , t_n)$ and their side information, there should exist coefficients $\alpha_{ij1}, \alpha_{ij2}, \cdots , \alpha_{ijn}$ for each message $w_{ij}$ ($1\le j\le l_i$) so that:
\begin{eqnarray}
\sum_{k=1}^n \alpha_{ijk} t_k \label{eqn:parta1}
\end{eqnarray}
is equal to $w_{ij}$ plus a linear combination of $w_{i'j'}$ that are available to node $i$ as side information, i.e.
\begin{eqnarray}
\sum_{k=1}^n \alpha_{ijk} t_k=w_{ij}+\sum_{i',j': (i, {i'})\in\mathcal{E}}w_{i'j'}\cdot \gamma_{i'j'},\label{eqn:parta2}
\end{eqnarray}
for some coefficients $\gamma_{i'j'}$. 

Now, we turn to the proof of the lemma. Without loss of generality, suppose that the vertices of $\mathsf{G'}$ are $m-|\mathcal{V}'|+1, m-|\mathcal{V}'|+2, \cdots, m$. 
Note that this assumption and the assumption that $\mathsf{C}$ is in the row echelon form do not contradict the generality together. One can simply label the vertices of $\mathsf{G}$ such that nodes in $\mathsf{G'}$ be labeled with $m-|\mathcal{V}'|+1, m-|\mathcal{V}'|+2, \cdots, m$ and then applies some elementary row operations to find $\mathsf{C}$ in row echelon form. 
The statement of the theorem basically asks us to show that there is no need for nodes in $\mathsf{G''}$ to know (as side information) any of the messages for nodes in $\mathsf{G'}$, i.e. $W_i, i\in \mathsf{G'}$. To show this, we first define a new encoding linear map $f'$ and then prove that it enables nodes in $\mathsf{G''}$ to recover their intended messages without any need to have access to $W_i, i\in \mathsf{G'}$. Nodes in $\mathsf{G'}$ are also shown to be still able to decode their messages with the encoding function $f'$ using their side information (nodes in $\mathsf{G'}$ do not know any of the messages of nodes in $\mathsf{G''}$ since $\mathsf{G'}$ does not have any outgoing edge). Thus, the edges between $\mathsf{G'}$ and $\mathsf{G''}$ can be removed.

\emph{Part 1: definition of a new linear encoding function $f'$:}
Let
\begin{eqnarray}
s = \min \{k\ |\ c_{ijk} = 0, \quad \forall 1\le i \le m-|\mathcal{V}'|, \quad 1 \le j \le l_i \}.
\end{eqnarray}
Fig. \ref{fig:echelon} clarifies the definition of the $s$.
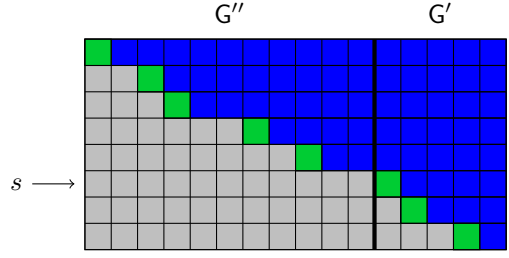
\begin{figure}\centering
\begin{tikzpicture}[scale = .7]
\filldraw[color=black, fill=gray!50!white](0,0) rectangle (8, 4);
\filldraw[color=black, fill=blue](0.0,3.5) rectangle (8, 4.0);
\filldraw[color=black, fill=blue](1.0,3.0) rectangle (8, 3.5);
\filldraw[color=black, fill=blue](1.5,2.5) rectangle (8, 3.0);
\filldraw[color=black, fill=blue](3.0,2.0) rectangle (8, 2.5);
\filldraw[color=black, fill=blue](4.0,1.5) rectangle (8, 2.0);
\filldraw[color=black, fill=blue](5.5,1.0) rectangle (8, 1.5);
\filldraw[color=black, fill=blue](6.0,0.5) rectangle (8, 1.0);
\filldraw[color=black, fill=blue](7.0,0.0) rectangle (8, 0.5);

\filldraw[color=black, fill=green!80!blue](0.0,3.5) rectangle (0.5, 4.0);
\filldraw[color=black, fill=green!80!blue](1.0,3.0) rectangle (1.5, 3.5);
\filldraw[color=black, fill=green!80!blue](1.5,2.5) rectangle (2.0, 3.0);
\filldraw[color=black, fill=green!80!blue](3.0,2.0) rectangle (3.5, 2.5);
\filldraw[color=black, fill=green!80!blue](4.0,1.5) rectangle (4.5, 2.0);
\filldraw[color=black, fill=green!80!blue](5.5,1.0) rectangle (6.0, 1.5);
\filldraw[color=black, fill=green!80!blue](6.0,0.5) rectangle (6.5, 1.0);
\filldraw[color=black, fill=green!80!blue](7.0,0.0) rectangle (7.5, 0.5);
\draw (0,0) rectangle (8, 4);
\draw(0.5, 0) -- (0.5, 4);
\draw(1.0, 0) -- (1.0, 4);
\draw(1.5, 0) -- (1.5, 4);
\draw(2.0, 0) -- (2.0, 4);
\draw(2.5, 0) -- (2.5, 4);
\draw(3.0, 0) -- (3.0, 4);
\draw(3.5, 0) -- (3.5, 4);
\draw(4.0, 0) -- (4.0, 4);
\draw(4.5, 0) -- (4.5, 4);
\draw(5.0, 0) -- (5.0, 4);
\draw(5.5, 0) -- (5.5, 4);
\draw(6.0, 0) -- (6.0, 4);
\draw(6.5, 0) -- (6.5, 4);
\draw(7.0, 0) -- (7.0, 4);
\draw(7.5, 0) -- (7.5, 4);

\draw(0, 0.5) -- (8, 0.5);
\draw(0, 1.0) -- (8, 1.0);
\draw(0, 1.5) -- (8, 1.5);
\draw(0, 2.0) -- (8, 2.0);
\draw(0, 2.5) -- (8, 2.5);
\draw(0, 3.0) -- (8, 3.0);
\draw(0, 3.5) -- (8, 3.5);

\draw [ultra thick] (5.5, 0) -- (5.5, 4);
\node at (6.75, 4.5){$\mathsf{G'}$};
\node at (2.75, 4.5){$\mathsf{G''}$};

\node at (-1.3, 1.25){$s$};
\draw [->] (-1, 1.25) -- (-0.2, 1.25);
\end{tikzpicture}
\caption{\footnotesize A pictorial representation of the row echelon form of $\mathsf{C}$ that clarifies the definition of $s$. Gray elements are zero and green elements are non-zero.\normalsize } \label{fig:echelon}
\end{figure}
Note that $\{k\ |\ c_{ijk}= 0, \quad \forall 1\le i \le m-|\mathcal{V}'|, \quad 1 \le j \le l_i \}$ cannot be empty. We prove this by contradiction. Suppose that the mentioned set is empty, then for each $1 \le k \le n$ there exists a $c_{ijk}\not = 0$ that $i \le m - |\mathcal{V}'|$. As a result, 
\begin{equation}
\forall 1 \le k \le n: \mathtt{i}_k \le m-|\mathcal{V}'|.
\label{eqn:ux}
\end{equation}
Now assume that ${\theta}$ is the smallest number that $\alpha_{m1{\theta}} \not = 0$, then:
\begin{align}
\sum_{k=1}^n \alpha_{m1k} \cdot t_k &= \sum_{k=1}^{{\theta}-1} \alpha_{m1k} \cdot t_k + \sum_{k={\theta}}^n \alpha_{m1k} \cdot t_k\\
&= \sum_{k=1}^{{\theta}-1} 0 \cdot t_k + \sum_{k={\theta}}^n \alpha_{m1k} t_k\\
&= \sum_{k={\theta}}^n \alpha_{m1k} t_k\\
&= \sum_{k={\theta}}^n \alpha_{m1k} \cdot(\sum_{p=1}^{m} \sum_{q=0}^{l_p} c_{pqk} \cdot w_{pq})
\end{align}
Note that the coefficient of $w_{\mathtt{i}_{\theta}\mathtt{j}_{\theta}}$ in the above statement is:
\begin{eqnarray}
\sum_{k={\theta}}^n \alpha_{m1k} \cdot c_{\mathtt{i}_{\theta}\mathtt{j}_{\theta}k}
\end{eqnarray}
Because $(\mathtt{i}_{\theta}, \mathtt{j}_{\theta})$ is lexicographically smaller than $(\mathtt{i}_{k}, \mathtt{j}_{k})$ for any $k > {\theta}$:
\begin{eqnarray}
\forall k > {\theta}: c_{\mathtt{i}_{{\theta}}\mathtt{j}_{{\theta}}k} = 0.
\end{eqnarray}
Then
\begin{eqnarray}
\sum_{k={\theta}}^n \alpha_{m1k} \cdot c_{\mathtt{i}_{\theta}\mathtt{j}_{\theta}k} = \alpha_{m1{\theta}} \cdot c_{\mathtt{i}_{\theta}\mathtt{j}_{\theta}{\theta}} \not = 0
\end{eqnarray}

Note that $\mathtt{i}_{\theta} \le m-|\mathcal{V}'|$ by eq. \eqref{eqn:ux}, so $\mathtt{i}_{\theta} \in \mathsf{G''}$. Hence, $w_{\mathtt{i}_{\theta}\mathtt{j}_{\theta}}$ is not provided as side information to node $m$, which is in $\mathsf{G'}$. As a result, the non-zero coefficient of $w_{\mathtt{i}_{\theta}\mathtt{j}_{\theta}}$ in $\sum_{k=1}^n \alpha_{m1k} \cdot t_k$ is in contradiction to eq. \eqref{eqn:parta2}. So we have proved that $\{k\ |\ c_{ijk}= 0, \quad \forall 1\le i \le m-|\mathcal{V}'|,\quad 1 \le j \le l_i \}$ is a non-empty set and thus $s$ is well-defined. 

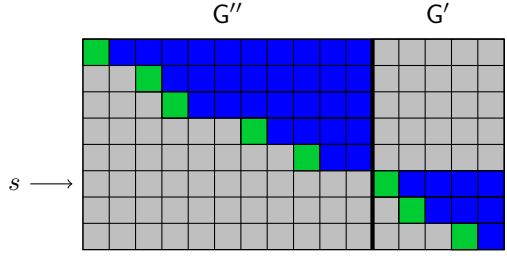
\begin{figure}\centering
\begin{tikzpicture}[scale=.7]
\filldraw[color=black, fill=gray!50!white](0,0) rectangle (8, 4);
\filldraw[color=black, fill=blue](0.0,3.5) rectangle (8, 4.0);
\filldraw[color=black, fill=blue](1.0,3.0) rectangle (8, 3.5);
\filldraw[color=black, fill=blue](1.5,2.5) rectangle (8, 3.0);
\filldraw[color=black, fill=blue](3.0,2.0) rectangle (8, 2.5);
\filldraw[color=black, fill=blue](4.0,1.5) rectangle (8, 2.0);
\filldraw[color=black, fill=blue](5.5,1.0) rectangle (8, 1.5);
\filldraw[color=black, fill=blue](6.0,0.5) rectangle (8, 1.0);
\filldraw[color=black, fill=blue](7.0,0.0) rectangle (8, 0.5);

\filldraw[color=black, fill=gray!50!white](5.5,1.5) rectangle (8, 4);a

\filldraw[color=black, fill=green!80!blue](0.0,3.5) rectangle (0.5, 4.0);
\filldraw[color=black, fill=green!80!blue](1.0,3.0) rectangle (1.5, 3.5);
\filldraw[color=black, fill=green!80!blue](1.5,2.5) rectangle (2.0, 3.0);
\filldraw[color=black, fill=green!80!blue](3.0,2.0) rectangle (3.5, 2.5);
\filldraw[color=black, fill=green!80!blue](4.0,1.5) rectangle (4.5, 2.0);
\filldraw[color=black, fill=green!80!blue](5.5,1.0) rectangle (6.0, 1.5);
\filldraw[color=black, fill=green!80!blue](6.0,0.5) rectangle (6.5, 1.0);
\filldraw[color=black, fill=green!80!blue](7.0,0.0) rectangle (7.5, 0.5);
\draw (0,0) rectangle (8, 4);
\draw(0.5, 0) -- (0.5, 4);
\draw(1.0, 0) -- (1.0, 4);
\draw(1.5, 0) -- (1.5, 4);
\draw(2.0, 0) -- (2.0, 4);
\draw(2.5, 0) -- (2.5, 4);
\draw(3.0, 0) -- (3.0, 4);
\draw(3.5, 0) -- (3.5, 4);
\draw(4.0, 0) -- (4.0, 4);
\draw(4.5, 0) -- (4.5, 4);
\draw(5.0, 0) -- (5.0, 4);
\draw(5.5, 0) -- (5.5, 4);
\draw(6.0, 0) -- (6.0, 4);
\draw(6.5, 0) -- (6.5, 4);
\draw(7.0, 0) -- (7.0, 4);
\draw(7.5, 0) -- (7.5, 4);

\draw(0, 0.5) -- (8, 0.5);
\draw(0, 1.0) -- (8, 1.0);
\draw(0, 1.5) -- (8, 1.5);
\draw(0, 2.0) -- (8, 2.0);
\draw(0, 2.5) -- (8, 2.5);
\draw(0, 3.0) -- (8, 3.0);
\draw(0, 3.5) -- (8, 3.5);

\draw [ultra thick] (5.5, 0) -- (5.5, 4);
\node at (6.75, 4.5){$\mathsf{G'}$};
\node at (2.75, 4.5){$\mathsf{G''}$};

\node at (-1.3, 1.25){$s$};
\draw [->] (-1, 1.25) -- (-0.2, 1.25);
\end{tikzpicture}
\caption{\footnotesize A schematic representation of $c'_{ijk}$. \normalsize } \label{fig:definition-c-prime}
\end{figure}

Further let
\begin{eqnarray}
c'_{ijk} = \left\{ \begin{matrix} 0 & {k < s \ \ \ \mbox{and} \ \ \ i > m - |\mathcal{V}'|} \\ c_{ijk} &\mbox{otherwise}\end{matrix}\right. \label{eqn:cprime-def}\label{eqn:eqaaaaaa}
\end{eqnarray}
and
\begin{eqnarray}
t'_k = \sum_{i=1}^{m} \sum_{j=1}^{l_i} c'_{ijk} \cdot w_{ij}.
\end{eqnarray}
for all $1\leq k \leq n$. Set
\begin{eqnarray}
f'(W_1, W_2, \cdots , W_m) = (t'_1, t'_2, \cdots , t'_n)
\end{eqnarray}
A schematic representation of $c'_{ijk}$ is given in Fig. \ref{fig:definition-c-prime}.
Observe that eq. \eqref{eqn:eqaaaaaa} implies that $t'_k=t_k$ for all $k\geq s$.

\emph{Part 2: showing that nodes in $\mathsf{G'}$ are able to decode their message by using $(t'_1, t'_2, \cdots , t'_n)$ and their side information:}
Consider the coefficients $\alpha_{ijk}$ for decoding of the original linear mapping given in eq. \eqref{eqn:parta1}. We claim for any $r \in \mathsf{G'}$ (i.e. $r > m - |\mathcal{V}'|$) that $\alpha_{rj1} = \cdots= \alpha_{rj(s-1)} = 0$ for every $1 \le j \le l_r$. This completes the proof since for every $r \in \mathsf{G'}$:
\begin{eqnarray}
\sum_{k=1}^n \alpha_{rjk} t_k = \sum_{k=s}^n \alpha_{rjk} t_k = \sum_{k=s}^n \alpha_{rjk} t'_k. \label{eqn:parta3}
\end{eqnarray}
Equations \eqref{eqn:parta2} and \eqref{eqn:parta3} illustrate that every node $r \in \mathsf{G'}$ is able to obtain $w_{rj}$ in the new coding scheme by calculating $\sum_{k=s}^n \alpha_{rjk} t'_k$.

We prove $\alpha_{rj1} = \cdots= \alpha_{rj(s-1)} = 0$ for every $1 \le j \le l_r$ by contradiction. Suppose that $x$ is the smallest index that $\alpha_{rjx} \not = 0$ and $x < s$. By the definition of $s$, $\mathtt{i}_x \le m-|\mathcal{V}'|$ ($\mathtt{i}_x \in \mathsf{G''}$). It is also clear that the definition of $x$ results in:
\begin{eqnarray}
\alpha_{rjy} = 0,
\end{eqnarray}
for any $y < x$. Because $(\mathtt{i}_k, \mathtt{j}_k), k=1,2,\cdots,n$ is strictly increasing:
\begin{eqnarray}
c_{\mathtt{i}_x\mathtt{j}_xy} = 0,
\end{eqnarray}
for all $y > x$. As
\begin{eqnarray}
\sum_{k=1}^n \alpha_{rjk} \cdot t_k =\sum_{k=1}^n \alpha_{rjk} \cdot (\sum_{p=1}^{m} \sum_{q=1}^{l_p} c_{pqk} \cdot w_{pq}),
\end{eqnarray}
The coefficient of $w_{\mathtt{i}_x,\mathtt{j}_x}$ in $\sum_{k=1}^n \alpha_{rjk} \cdot t_k$ is:
\begin{align}
\sum_{k=1}^n \alpha_{rjk} \cdot c_{\mathtt{i}_x\mathtt{j}_xk}&= \sum_{k=1}^{x-1} \alpha_{rjk} \cdot c_{\mathtt{i}_x\mathtt{j}_xk} + \alpha_{rjx} \cdot c_{\mathtt{i}_x\mathtt{j}_xx}
+ \sum_{k=x+1}^n \alpha_{rjk} \cdot c_{\mathtt{i}_x\mathtt{j}_xk}\\
&= \sum_{k=1}^{x-1} 0 \cdot c_{\mathtt{i}_x\mathtt{j}_xk} + \alpha_{rjx} \cdot c_{\mathtt{i}_x\mathtt{j}_xx} 
+ \sum_{k=x+1}^n \alpha_{rjk} \cdot 0\\
&= \alpha_{rjx} \cdot c_{\mathtt{i}_x\mathtt{j}_xx} \\&\not = 0,
\end{align}
which is in contradiction to the independency of $\sum_{k=1}^n \alpha_{rjk} t_k$ from $w_{\mathtt{i}_x\mathtt{j}_x}$, that was guaranteed by eq. \eqref{eqn:parta2}. (Note that $r \in \mathsf{G'}$ and $\mathtt{i}_x \in \mathsf{G''}$, so $w_{\mathtt{i}_x\mathtt{j}_x}$ is not provided as side information to $r$)\\

\emph{Part 3: showing that under $f'$ decoding is possible without the need for nodes in $\mathsf{G''}$ to know messages for nodes in $\mathsf{G'}$:}

For every $i \in \mathsf{G''}$, let
\begin{eqnarray}
\beta_{ijk} = \left\{ \begin{matrix} 0 & \qquad{k \ge s}; \\ \alpha_{ijk} &\qquad {k < s}. \end{matrix}\right.
\end{eqnarray}
We claim that for every $i \in \mathsf{G''}$:
\begin{eqnarray}
\sum_{k=1}^n \beta_{ijk} t'_k=w_{ij}+\sum_{i',j': (i, {i'})\in\mathcal{E} \footnotesize{\mbox{ and }} i' \not \in \mathsf{G'}}w_{i'j'}\cdot \gamma_{i'j'},
\label{eqn:eqaabbcc}
\end{eqnarray}
where $\gamma_{i'j'}$ is given in eq. \eqref{eqn:parta2}. 
This shows that nodes at $\mathsf{G''}$ are able to decode their messages using $(t'_1, t'_2, \cdots, t'_n)$ and their side information in $\mathsf{G''}$ (excluding side information from nodes at $\mathsf{G'}$). We have:
\begin{align}
\sum_{k=1}^n \beta_{ijk} t'_k &= \sum_{k=1}^{s-1} \beta_{ijk} t'_k + \sum_{k=s}^n \beta_{ijk} t'_k\\
&= \sum_{k=1}^{s-1} \alpha_{ijk} t'_k + \sum_{k=s}^n 0 \cdot t'_k\\
&= \sum_{k=1}^{s-1} \alpha_{ijk} \cdot (\sum_{p=1}^{m} \sum_{q=0}^{l_p} c'_{pqk} \cdot w_{pq})\\
&= \sum_{k=1}^{s-1} \alpha_{ijk} \cdot(\sum_{p=1}^{m-|\mathcal{V}'|} \sum_{q=0}^{l_p} c'_{pqk} \cdot w_{pq})\label{eqn:part3-1} \\ 
&= \sum_{k=1}^{s-1} \alpha_{ijk} \cdot(\sum_{p=1}^{m-|\mathcal{V}'|} \sum_{q=0}^{l_p} c_{pqk} \cdot w_{pq})\label{eqn:part3-2},
\end{align}
where eqs. \eqref{eqn:part3-1} and \eqref{eqn:part3-2} follow from the definition of $c'_{ijk}$ in eq. \eqref{eqn:cprime-def}. Note that the expression of eq. \eqref{eqn:part3-2} does not include any of $w_{ij}$ for $i > m - |\mathcal{V}'|$. Moreover, the coefficient of $w_{ij}$ for $i\leq m-|\mathcal{V}'|$ are the same as those in $\sum_{k=1}^n \alpha_{ijk} t_k$. This establishes eq. \eqref{eqn:eqaabbcc}.
\end{proof}

\subsubsection{Proof of part (b)}\-

The proof has two parts: first we show that 
$\bigcup_{\alpha\in[0,1]} \alpha \mathcal{C}'\oplus (1-\alpha)\mathcal{C}'' \subseteq \mathcal{C}$
and then we will finish the proof by showing that 
$\mathcal{C}\subseteq\bigcup_{\alpha\in[0,1]} \alpha \mathcal{C}'\oplus (1-\alpha)\mathcal{C}''$.

Before starting the proof, let us label the vertices of $\mathsf{G}$ so that the vertices of $\mathsf{G}'$ come first.

\emph{Proving $\bigcup_{\alpha\in[0,1]} \alpha \mathcal{C}'\oplus (1-\alpha)\mathcal{C}'' \subseteq \mathcal{C}$}:
Take an arbitrary vector $\overline{\bf r}'$ in $\mathcal{C}'$. Then we can allocate all of our resources for $\mathsf{G}'$ and do not send anything for $\mathsf{G}''$. This shows that $(\overline{\bf r}', 0)$ is in $\mathcal{C}$. Similarly for any $\overline{\bf r}''$ in $\mathcal{C}''$, we have that $(0, \overline{\bf r}'')$ is in $\mathcal{C}$. Using the standard time-sharing techniques, one can show that the capacity region of the index coding problem is a convex set. Therefore for any $\alpha\in[0,1]$ the rate $(\alpha \overline{\bf r}', (1-\alpha)\overline{\bf r}'')\in \mathcal{C}$. This completes the proof. 

\emph{Proving $\mathcal{C} \subseteq \bigcup_{\alpha\in[0,1]} \alpha \mathcal{C}'\oplus (1-\alpha)\mathcal{C}''$}:
For any rate vector $\overline{\bf r}\in\mathcal{C}$, there exist a sequence of codes like $C_1, C_2, \cdots$ whose rates converge to $\overline{\bf r}$. Take some $\epsilon>0$ and a code described by encoding function $f$ whose rate $\overline{\bf r}_{\epsilon}$ is within $\epsilon$ distance of $\overline{\bf r}$.
\begin{itemize}
\item \textbf{Linear Case}: Suppose that 
$f:\mathcal{W}_{1}\times \mathcal{W}_{2}\times \cdots \times\mathcal{W}_{|\mathcal{V}'|+|\mathcal{V}''|} \rightarrow \mathbb{F}^{n}$.
In the proof of the part (a), we showed that there exist encoding functions 
$f':\mathcal{W}_{1}\times \mathcal{W}_{2}\times \cdots \times\mathcal{W}_{|\mathcal{V}'|} \rightarrow \mathbb{F}^{n'}$
 and $f'':\mathcal{W}_{|\mathcal{V}'|+1}\times \mathcal{W}_{|\mathcal{V}'|+2}\times \cdots \times\mathcal{W}_{|\mathcal{V}'| + |\mathcal{V}''|} \rightarrow \mathbb{F}^{n''}$ which are respectively valid for $\mathsf{G}'$ and $\mathsf{G}''$. Additionally, the size of the range of the concatenation of $f'$ and $f''$ equals the size of the range of $f$, i.e. $n = n' + n''$. Hence, if we call the rates of $f'$ and $f''$, $\overline{\bf r}'$ and $\overline{\bf r}''$, we will have:
\begin{eqnarray}
\overline{\bf r}_{\epsilon} &=& \left(\frac{\log_\mathbb{F}|\mathcal{W}_{1}|}{n'+n''}, \frac{\log_\mathbb{F}|\mathcal{W}_{2}|}{n'+n''}, \cdots, \frac{\log_\mathbb{F}|\mathcal{W}_{|\mathcal{V}'|+|\mathcal{V}''|}|}{n'+n''}\right)	\\
&=& \biggr(\frac{n'}{n'+n''}\left(\frac{\log_\mathbb{F}|\mathcal{W}_{1}|}{n'}, \cdots, \frac{\log_\mathbb{F}|\mathcal{W}_{|\mathcal{V}'|}|}{n'}\right), \nonumber\\
&&\ \ \frac{n''}{n'+n''}\left(\frac{\log_\mathbb{F}|\mathcal{W}_{|\mathcal{V}'|+1}|}{n''}, \cdots, \frac{\log_\mathbb{F}|\mathcal{W}_{|\mathcal{V}'|+|\mathcal{V}''|}|}{n''}\right)\biggr)	\\
&=& \biggr(\frac{n'}{n'+n''}\overline{\bf r}', (1-\frac{n'}{n'+n''})\overline{\bf r}''\biggr)
\end{eqnarray}

Since $\overline{\bf r}'\in \mathcal{C}'$ and $\overline{\bf r}''\in \mathcal{C}''$, above statement results in the fact that $\overline{\bf r}_{\epsilon}$ lies in $\bigcup_{\alpha\in[0,1]} \alpha \mathcal{C}'\oplus (1-\alpha)\mathcal{C}''$. By the definition of the asymptotic capacity region, $\mathcal{C}'$ and $\mathcal{C}''$ are closed sets. We are done with the proof by noting that $\overline{\bf r}_{\epsilon}$  can be made arbitrarily close to $\overline{\bf r}$.

\item \textbf{Non-linear Case}: Suppose that 
$f:\mathcal{W}_{1}\times \mathcal{W}_{2}\times \cdots \times\mathcal{W}_{|\mathcal{V}'|+|\mathcal{V}''|} \rightarrow \{1, 2, \cdots, N\}$.
In the proof of the part (a), we showed that we can find encoding functions 
$f':\mathcal{W}_{1}^n\times \mathcal{W}_{2}^n\times \cdots \times\mathcal{W}_{|\mathcal{V}'|}^n \rightarrow \{1, 2, \cdots, 2^{n(K'+\delta)}\}$
 and $f'':\mathcal{W}_{|\mathcal{V}'|+1}^n\times \mathcal{W}_{|\mathcal{V}'|+2}^n\times \cdots \times\mathcal{W}_{|\mathcal{V}'| + |\mathcal{V}''|}^n \rightarrow \{1, 2, \cdots, 2^{n(K''+\delta)}\}$ in which $K'$ and $K''$ satisfy $N = 2^{K'+K''}$, $\delta$ is an arbitrary positive real number, and an appropriate $n$ can be found for any fixed $\delta$ so that such functions exist. Moreover, $f'$ and $f''$ are respectively valid for $\mathsf{G}'$ and $\mathsf{G}''$ over the alphabet sets of 
$\mathcal{W}_{1}^n, \mathcal{W}_{2}^n, \cdots, \mathcal{W}_{|\mathcal{V}'|}^n$ and 
$\mathcal{W}_{|\mathcal{V}'|+1}^n, \mathcal{W}_{|\mathcal{V}'|+2}^n, \cdots, \mathcal{W}_{|\mathcal{V}'| + |\mathcal{V}''|}^n$. Hence, if we call the rates of $f'$ and $f''$, $\overline{\bf r}'$ and $\overline{\bf r}''$, we will have:
\begin{eqnarray}
\overline{\bf r}_{\epsilon}&=& \left(\frac{\log|\mathcal{W}_{1}|}{\log N}, \frac{\log|\mathcal{W}_{2}|}{\log N}, \cdots, \frac{\log|\mathcal{W}_{|\mathcal{V}'|+|\mathcal{V}''|}|}{\log N}\right)	\\
&=& \left(\frac{\log|\mathcal{W}_{1}^n|}{n\log N}, \frac{\log|\mathcal{W}_{2}^n|}{n\log N}, \cdots, \frac{\log|\mathcal{W}_{|\mathcal{V}'|+|\mathcal{V}''|}^n|}{n\log N}\right)	\\
&=& \left(\frac{\log|\mathcal{W}_{1}^n|}{n(K'+K'')}, \frac{\log|\mathcal{W}_{2}^n|}{n(K'+K'')}, \cdots, \frac{\log|\mathcal{W}_{|\mathcal{V}'|+|\mathcal{V}''|}^n|}{n(K'+K'')}\right)	\\
&=& \Biggr(\frac{K'+\delta}{K'+K''}\left(\frac{\log|\mathcal{W}_{1}^n|}{n(K'+\delta)}, \cdots, \frac{\log|\mathcal{W}_{|\mathcal{V}'|}^n|}{n(K'+\delta)}\right), \nonumber\\
&&\ \ \frac{K''+\delta}{K'+K''}\left(\frac{\log|\mathcal{W}_{|\mathcal{V}'|+1}^n|}{n(K''+\delta)}, \cdots, \frac{\log|\mathcal{W}_{|\mathcal{V}'|+|\mathcal{V}''|}^n|}{n(K''+\delta)}\right)\Biggr)\\
&=&\left(\frac{K'+\delta}{K'+K''} \overline{\bf r}', \frac{K''+\delta}{K'+K''}\overline{\bf r}''\right)\\
&=&\frac{K'+K''+2\delta}{K'+K''}\left(\frac{K'+\delta}{K'+K''+2\delta} \overline{\bf r}', \frac{K''+\delta}{K'+K''+2\delta}\overline{\bf r}''\right)\\
&=&\frac{K'+K''+2\delta}{K'+K''}\left(\frac{K'+\delta}{K'+K''+2\delta} \overline{\bf r}', (1-\frac{K'+\delta}{K'+K''+2\delta})\overline{\bf r}''\right)
\end{eqnarray}
Thus,
\begin{equation}
\frac{K'+K''}{K'+K''+2\delta} \overline{\bf r}_{\epsilon} = \left(\frac{K'+\delta}{K'+K''+2\delta} \overline{\bf r}', (1-\frac{K'+\delta}{K'+K''+2\delta})\overline{\bf r}''\right)
\end{equation}
As $\overline{\bf r}' \in \mathcal{C}'$ and $\overline{\bf r}'' \in \mathcal{C}''$, $\frac{K'+K''}{K'+K''+2\delta}\overline{\bf r}_{\epsilon}$ lies in $\bigcup_{\alpha\in[0,1]} \alpha \mathcal{C}'\oplus (1-\alpha)\mathcal{C}''$ for any $\delta > 0$. Since $\bigcup_{\alpha\in[0,1]} \alpha \mathcal{C}'\oplus (1-\alpha)\mathcal{C}''$ is a closed set and we can make $\delta$ and $\epsilon$ as close to zero as we want, we will be done.
\end{itemize} \mbox{}\hspace*{\fill}\nolinebreak\mbox{$\rule{0.6em}{0.6em}$}

\subsection{Proof of Theorem \ref{thm2}}
We prove two parts of Theorem \ref{thm2} in the following two subsections.
\subsubsection{Proof of part (a)} \-

{\sl Linear one-shot case:} If $\mathsf{G}$ is USCS, then proof is finished. Otherwise, $\mathsf{G}$ contains an edge like $e = (u, v)$ which is not located in any cycles. Let $\mathcal{V}_1$ be the set of vertices that can be reached from $v$. Morever, let $\mathcal{V}_2$ be the set of vertices who cannot be reached from $v$. It is easy to verify that there will be no edge that starts from $\mathcal{V}_1$ and finishes in $\mathcal{V}_2$. Using part (a) of Theorem \ref{thm3new}, we can remove all edges between $\mathcal{V}_1$ and $\mathcal{V}_2$ including $e$, so that the rate region does not shrink. As the number of the edges of $\mathsf{G}$ is finite, by repeating this process we can find a USCS subgraph of $\mathsf{G}$ like $\mathsf{G}'$ whose rate region equals the rate region of $\mathsf{G}$. Hence, if $\mathsf{G}$ is a critical graph, it should be equal to $\mathsf{G}'$ which is USCS. In other words, any critical graph for one-shot linear index coding is USCS.

The proof for {\sl Non-linear asymptotic index coding} using part (a) of Theorem \ref{thm3new} is similar.

{\sl Linear asymptotic case:} This follows from the one-shot case. For any code $C$ of length $n$, there is another code $C'$ with the same rate vector on a subgraph $\mathsf{G}'$ of $\mathsf{G}$ that is USCS. Now, given any arbitrary sequence of codes $C_1, C_2, \cdots$ whose rate vector converges to a given rate vector $\overline{\bf r} = (r_1, r_2, \cdots, r_m)$, we can find a sequence of codes $C'_1, C'_2, \cdots$ on subgraphs $\mathsf{G}'_1, \mathsf{G}'_2, \cdots$ whose rate vector converges to the same rate vector $\overline{\bf r} = (r_1, r_2, \cdots, r_m)$. Since any graph $\mathsf{G}$ has only a finite number of subgraphs, we can find indices $\mathtt{i}_1<\mathtt{i}_2<\cdots$ such that $\mathsf{G}'_{\mathtt{i}_k}=\tilde{\mathsf{G}}$ are identical. The subsequence of the codes $C'_{\mathtt{i}_k}$ is defined on the USCS graph $\tilde{\mathsf{G}}$ and has a rate vector that converges to $\overline{\bf r} = (r_1, r_2, \cdots, r_m)$. Since $\mathsf{G}$ and is critical and $\tilde{\mathsf{G}}$ is a subgraph of $\mathsf{G}$, we conclude that $\mathsf{G}=\tilde{\mathsf{G}}$ implying that $\mathsf{G}$ is USCS.\\
 \mbox{}\hspace*{\fill}\nolinebreak\mbox{$\rule{0.6em}{0.6em}$}

\subsubsection{Proof of part (b)}
To prove this part we need to show that a critical graph exists for one-shot non-linear case  that is not USCS. 

Consider the graph given in Fig. \ref{Fig5}. We call this graph $\mathsf{G}=(\mathcal{V}, \mathcal{E})$.
 Assume that
\begin{align*}
\mathcal{W}_i&= \{0, 1\},\qquad 1\leq i \leq 5,\\
\mathcal{W}_6&= \{0, 1, 2, 3, 4\}.
\end{align*}
We have the following claim:

\begin{claim}Sending a symbol from $\{1,2,\cdots, 32\}$ as the public message suffices for every node to decode its message.
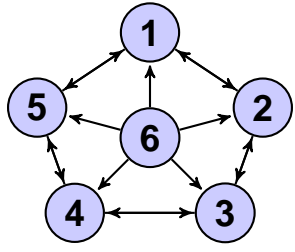
\begin{figure}\centering
\begin{tikzpicture}
[->, >=stealth', shorten >=1pt,auto,node distance=1.7cm, thick, main node/.style={circle,fill=blue!20,draw,font=\sffamily\Large\bfseries}]
\node[main node] (6)at(0,0) {6};
\node[main node] (1)at(0, 1.4)  {1};
\node[main node] (2) at(1.5, 0.4) {2};
\node[main node] (5) at(-1.5, 0.4){5};
\node[main node] (4) at(-1, -1) {4};
\node[main node] (3) at(1, -1) {3};
\path[every node/.style={font=\sffamily\small}]
(1) edge node {} (2)
edge node {} (5)
(2) edge node {} (1)
edge node {} (3)
(3) edge node {} (2)
edge node {} (4)
(4) edge node {} (3)
edge node {} (5)
(5) edge node {} (1)
edge node {} (4)
(6) edge node {} (1)
edge node {} (2)
edge node {} (3)
edge node {} (4)
edge node {} (5);
\end{tikzpicture}
\caption{\footnotesize In the index coding problem associated with this graph, removing the edges which belong to no cycle implies a larger public message rate.\normalsize \label{Fig5}}
\end{figure}
However, if we remove the edges connected to node $6$, which do not belong to any cycle, we need at least 35 symbol to have a successful transmission of the messages. 
\end{claim}
This claim establishes the desired result, since if $\mathsf{G}$ is critical it would be an instance of a non-USCS graph that is critical. If $\mathsf{G}$ is not critical, there is a subgraph $\mathsf{G}'$ of it (obtained by removing edges from $\mathsf{G}$) that is critical; that is the graph $\mathsf{G}'$ is such that sending a symbol from $\{1,2,\cdots, 32\}$ as the public message suffices for every node to decode its message. However any further removal of edges from $\mathsf{G}'$ results in a graph that does not have this property. By the above claim, the minimal graph $\mathsf{G}'$ should contain at least one of the edges connected to the node $6$; since if not, $\mathsf{G}'$ would be a subgraph of the graph shown in the claim to need at least 35 symbols. Therefore, $\mathsf{G}'$ contains an edge that is not on any cycle. Hence it is a non-USCS and critical graph.

We now turn to the proof of the claim. In order to construct the coding scheme using 32 symbols for $\mathsf{G}$, first note that $W_1W_2W_3W_4W_5$ forms a binary sequence of the length 5. Based on the value of $W_6$, we XOR this sequence with one the following sequence: $00000$, $10001$, $01111$, $01100$, $10111$, that is, if $W_6$ is 0 we XOR the sequence with $00000$, if it is 1 we XOR it with $10001$, and so on. Then, we transmit the result as the public message (the public message has 32 different possibilities and can be transmitted). Let us denote the 5-bit public message by $\tilde{W}_1\tilde{W}_2\tilde{W}_3\tilde{W_4}\tilde{W}_5$. It is sufficient to show that every node can decode its message with the help of the public message and its side information. First of all, because the node $6$ knows the message of 1 to 5, it can XOR their message by the public message and from the XOR decode its message. For the other nodes, note that $W_i\oplus \tilde{W}_i$ for $1\leq i \leq 5$ is a function of the side information of node $i$, and therefore, node $i$ can decode its message. We explain the decoding process for node $1$; the decoding process for other nodes is similar. Node $1$ knows $W_2$ and $W_5$ . By comparing these two bits with $\tilde{W}_2$ and $\tilde{W}_5$, node $1$ can exactly recover $W_6$ if it is equal to $0$, $2$ or $3$. If $W_6$ is equal to $1$ or $4$, node $1$ cannot find the exact value of $W_6$. However in both cases of $W_6=1,4$ we have $\tilde{W}_1=\neg W_1$, and by flipping $\tilde{W}_1$ the first node can recover its intended bit.

In order to prove that if we remove the edges connected to node $6$, at least 35 symbols are needed, suppose that there exists a coding scheme which requires at most 34 symbols. According to Pigeonhole Principle, we can conclude that there exists $w_6\in\{0,1,2,3,4\}$ so that if $W_6=w_6$, public message gets at most 6 distinct different values when we vary the $w_1, w_2, \cdots, w_5$, i.e. the cardinality of the set
$$\big\{f(w_1, w_2, \cdots, w_6): w_1,w_2, \cdots, w_5\in \{0,1\}\big\}$$
is at most 6. For this value of $w_6$, consider the following function over five variables $w_1, w_2, \cdots, w_5$:
$$\tilde{f}(w_1, w_2, \cdots, w_5)=f(w_1, w_2, \cdots, w_5, w_6).$$
Since $W_6$ was independent of $(W_1, \cdots, W_5)$ and we have zero probability of error, the function $\tilde{f}$ is a valid encoding function for a cycle of length 5. This contradicts Lemma \ref{lemma:cycle-five} below. \mbox{}\hspace*{\fill}\nolinebreak\mbox{$\rule{0.6em}{0.6em}$}
\begin{lemma}
\label{lemma:cycle-five}
The bidirectional cycle of length 5 with $\mathcal{W}_i=\{0,1\}$ needs a public message of alphabet size 7 to achieve a zero probability of error for the one-shot problem.
\end{lemma}
\begin{proof}
We prove this lemma by contradiction. Assume otherwise that there exists a coding scheme that uses a public message with 6 possibilities. From the Pigeonhole Principle, we conclude that the encoding function maps at least 6 combinations of the messages to one symbol, i.e. there are six sequences of $(w_{1i}, w_{2i}, \cdots, w_{5i})\in\{0,1\}^5$, $i=1,2,\cdots, 6$, whose $f(w_{1i}, w_{2i}, \cdots, w_{5i})$ are equal, i.e. their corresponding public message is the same. Thus, the nodes should be able to recover their own messages using their side information. In other words, for instance for node $1$, if $w_{1i}\neq w_{1i'}$ for some $i$ and $i'$, then we should have $(w_{2i}, w_{5i})\neq (w_{2i'}, w_{5i'})$. Thus the six sequences should be distinguishable, where we call two sequences $(w_{1}, w_{2}, \cdots, w_{5})$ and $(w'_{1}, w'_{2}, \cdots, w'_{5})$ distinguishable if for each $i\in\{1,2,\cdots, 5\}$ either $w_i=w'_i$ or the $w_{j}\neq w'_j$ for some $j: (i,j)\in\mathcal{E}$.

Given a sequence $(w_{1}, w_{2}, \cdots, w_{5})$, consider the graph induced on the set of vertices $\{j: w_j=1\}$. We call the sequence $(w_{1}, w_{2}, \cdots, w_{5})$ ``good" if the induced graph does not contain of a vertex of degree zero (i.e. is connected). For instance, in a cycle of size $5$ if we take $(w_{1}, w_{2}, \cdots, w_{5})=(1,1,1,0,0)$, the induced graph would be on nodes $1,2,3$ which is connected. However $(w_{1}, w_{2}, \cdots, w_{5})=(1,1,0,1,0)$ corresponds to the induced graph on nodes $1,2,4$ which is not connected since node $4$ is not connected to nodes $1$ and $2$.
It is easy to verify that $(w_{1}, w_{2}, \cdots, w_{5})$ and $(w'_{1}, w'_{2}, \cdots, w'_{5})$ distinguishable if and only if their bitwise XOR is good. For instance $(1,1,0,1,0)$ and $(0,0,0,0,0)$ are not distinguishable (by node 4) since their XOR, $(1,1,0,1,0)$ is not good.

Now, we know that the XOR of any two of $(w_{1i}, w_{2i}, \cdots, w_{5i})$, $i=1,2,\cdots, 6$ is good. We show that this cannot happen. Without loss of generality, we can assume that one of the six sequences is the all zero sequence. Therefore, we should look for 5 sequences that are individually good, and their pairwise bitwise XOR is also good. In Appendix \ref{appendixB}, we provide a code in $C++$ which checks all possible cases and shows that such a set of sequences does not exist. 
\end{proof}
\subsection{Proof of Theorem \ref{thm:newstr1}}
Suppose that $\mathsf{G}=(\mathcal{V}, \mathcal{E})$ is a bidirectional graph where $\mathcal{V} = \{1, 2, \cdots, m\}$. To show that $\mathsf{G}$ is critical, we need to find a rate vector, $\overline{\bf r} = (r_1, r_2, \cdots, r_m)$, for every $e = (u, v)\in \mathcal{E}$ that is achievable in $\mathsf{G}$, but  it is not achievable in $\mathsf{G}-e$. We define $\overline{\bf r}$ in the following manner:
\\
\[ r_i = \left\{ 
  \begin{array}{l l}
    1 & \quad \text{$if$ $i=u$ $or$ $i=v$}\\
    0 & \quad \text{$otherwise$}
  \end{array} \right.\]
To show that $\overline{\bf r}$ is achievable in $\mathsf{G}$, suppose that $W_i$ is the message of node $i$, and $W_i \in \mathcal{W}_i$ where:
\[ \left\{ 
  \begin{array}{l l}
    \mathcal{W}_i = \{0, 1\} & \quad \text{if $i=u$ or $i=v$}\\
    \mathcal{W}_i = \{0\} & \quad \text{$otherwise$}
  \end{array} \right.\]
 Now, if we send $W_u \oplus W_v$ as public message, then $u$ and $v$ can decode their message, because they have the message of each other as side information and the sum of their message. As $\mathcal{W}_i$ has only one element for $i \neq u, v$, the other vertices can trivially decode their message. Therefore, $\overline{\bf r}$ is supported by $\mathsf{G}$.

Additionally, since the set $\{u, v\}$ in $\mathsf{G} - e$ has no directed cycle, Lemma \ref{lemma:rates-sum}  implies that for every ${\overline{\bf r'}} = (r'_1, r'_2, \cdots, r'_m)$ supported by $\mathsf{G} -e$, $r_u'+r_v'\leq 1$. Thus, $\overline{\bf r}$ cannot be supported by $\mathsf{G}-e$.

Next we show that a cycle of size four with vertices $\{1, 2, 3, 4\}$ and edges $\{(1, 2), (2, 1), (2, 3), (3, 2), (3, 4), (4, 3),$ $(4, 1), (1, 4)\})$ is not symmetric rate critical. If this graph supports the rate vector $\overline{\bf r} = (r, r, r, r)$, as the set $\{1, 3\}$ has no directed cycle, Lemma \ref{lemma:rates-sum} gives that $r\leq \frac{1}{2}$. In addition, consider the subgraph $\mathsf{H}$ of the cycle with edges$\{(1, 2), (2, 1), (3, 4), (4, 3)\}$. If we send two bits $(W_1 \oplus W_2, W_3\oplus W_4)$ as public message, then all nodes can decode their message. Hence, the rate $(\frac{1}{2}, \frac{1}{2}, \frac{1}{2}, \frac{1}{2})$ is achievable in $\mathsf{H}$. Now, since removing edges $(2, 3), (3, 2), (4, 1), (1, 4)$ do not change the symmetric capacity region of cycle of size four, it is not symmetric rate critical.
\mbox{}\hspace*{\fill}\nolinebreak\mbox{$\rule{0.6em}{0.6em}$}

\subsection {Proof of Theorem \ref{thm3}}
\subsubsection{Criticality of $\mathsf{G} \cup \mathsf{H}$}
In order to show the criticality of $\mathsf{G} \cup \mathsf{H}$ we need to show that by eliminating every edge like $e$ from $\mathsf{G} \cup \mathsf{H}$, the capacity region of the index coding problem related to $\mathsf{G} \cup \mathsf{H}$ shrinks strictly.
Without loss of generality assume that $e \in \mathcal{E}_{\mathsf{G}}$. As $\mathsf{G}$ is a critical graph, there exists a rate vector like $\overline{\bf r}$ that supports $\mathsf{G}$, but not $\mathsf{G}'(\mathcal{V}_{\mathsf{G}}, \mathcal{E}_{\mathsf{G}} - \{e\})$. Now, consider a rate vector for the index coding problem introduced by $\mathsf{G} \cup \mathsf{H}$ in which the rates of nodes in $\mathsf{H}$ are all zero and rates of the nodes in $\mathsf{G}$ equals $\overline{\bf r}$. This rate vector is evidently admissible for $\mathsf{G} \cup \mathsf{H}$, but not for $\mathsf{G}' \cup \mathsf{H}$ (which is $\mathsf{G} \cup \mathsf{H}$ after elimination of $e$).
\subsubsection{Symmetric Criticality of $\mathsf{G} \cup \mathsf{H}$ in Asymptotic Scenarios}
Showing that the maximal symmetric rate also reduces after we remove an edge from $\mathsf{G} \cup \mathsf{H}$ is more challenging. 
Let $r_1$ and $r_2$ be the maximal symmetric rate for $\mathsf{G}$ and $\mathsf{H}$ respectively. It is clear that concatenation of these two coding functions with proportion of $\frac{r_2}{r_1+r_2}$ and $\frac{r_1}{r_1+r_2}$ for $\mathsf{G}$ and $\mathsf{H}$ respectively, results in a coding function for $\mathsf{G} \cup \mathsf{H}$ with the symmetric rate of $r = \frac{r_1r_2}{r_1+r_2}$. 

We claim that this symmetric rate would not be achievable if any edge like $e$ is removed from $\mathsf{G} \cup \mathsf{H}$. This will prove the symmetric criticality of $\mathsf{G} \cup \mathsf{H}$. We are going to prove this claim by contradiction. 
Without loss of generality, assume that $e$ is an edge of $\mathsf{G}$. We refer to the graph obtained by $\mathsf{G}$ after elimination of $e$ as $\mathsf{G}'$. Suppose that there exists a coding function like $f$ for $\mathsf{G}' \cup \mathsf{H}$ with the symmetric rate of $r$. From Theorem \ref{thm3new} then, there exist some $\alpha\in[0,1]$ and symmetric rates $r'_1$ and $r'_2$ for $\mathsf{G}'$ and $\mathsf{H}$ such that $r=\alpha r'_1=\bar{\alpha}r'_2$. This implies that $$\frac{r'_1r'_2}{r'_1 + r'_2}=
\frac{
\frac{r^2}{\alpha\bar{\alpha}}}{\frac{r}{\alpha}+\frac{r}{\bar{\alpha}}}=r=\frac{r_1r_2}{r_1+r_2}$$
Thus,
\begin{align}
&\frac{1}{r'_1} + \frac{1}{r'_2} =\frac{1}{r_1} + \frac{1}{r_2}	\label{eqn:rates_thm4}
\end{align}
However, by the symmetric criticality of $\mathsf{G}$ and by the definitions of $r_1$ and $r_2$, we have that 
\begin{eqnarray}
r'_1 < r_1,&\quad &r'_2 \le r_2\nonumber	\\
\Rightarrow \frac{1}{r_1} < \frac{1}{r'_1},&\quad& \frac{1}{r_2} \le \frac{1}{r'_2}\nonumber\\
\Rightarrow  \frac{1}{r'_1} + \frac{1}{r'_2} & > &  \frac{1}{r_1} + \frac{1}{r_2}	\label{eqn:rates_thm4_2}
\end{eqnarray}
Equations \eqref{eqn:rates_thm4} and \eqref{eqn:rates_thm4_2} are in contradiction with each other. This contradiction completes the proof.

\subsubsection{Symmetric Criticality of $\mathsf{G} \cup \mathsf{H}$ in One-Shot Linear Scenario}
Consider the symmetric one-shot linear index coding problems defined over $\mathsf{G}$, $\mathsf{H}$, and $\mathsf{G} \cup \mathsf{H}$. Assume that the alphabet of each node in each of these problems is $\mathbb{F}^l$ for some finite field $\mathbb{F}$. Let $n_1$ be the minimum possible positive integer number such that there exists a valid linear coding function with the output size of $n_1$ symbols over $\mathbb{F}$. Define $n_2$ for $\mathsf{H}$ in the same manner. It is clear that there exists an encoding function for the problem related to $\mathsf{G}\cup \mathsf{H}$ that uses a public message of $n_1 + n_2$ symbols by concatenation of the encoding functions that use $n_1$ symbols for $\mathsf{G}$ and $n_2$ symbols for $\mathsf{H}$. Hence, the symmetric rate of $\frac{l}{n_1+n_2}$ is achievable for the index coding problem introduced by $\mathsf{G} \cup \mathsf{H}$.

We are going to show that if $\mathsf{G}$ and $\mathsf{H}$ are both symmetric critical, then $\mathsf{G} \cup \mathsf{H}$ is symmetric critical too. To prove the symmetric criticality of $\mathsf{G} \cup \mathsf{H}$, we will prove that any valid coding function for $\mathsf{G} \cup \mathsf{H}$ needs at least $n_1 + n_2 + 1$ public symbols after removal of any edge like $e$ from $\mathsf{G} \cup \mathsf{H}$. Without loss of generality, we assume that $e$ is removed from the $\mathsf{G}$ component of $\mathsf{G} \cup \mathsf{H}$. We refer to the graph obtained by $\mathsf{G}$ after the removal of $e$ as $\mathsf{G}'$. Then, the graph obtained from $\mathsf{G} \cup \mathsf{H}$ after the removal of $e$ would be $\mathsf{G}' \cup \mathsf{H}$. Let $f$ be a valid encoding function for $\mathsf{G}' \cup \mathsf{H}$, we have shown in the proof of part (a) of Theorem \ref{thm3new} that there exist two valid encoding functions $f'$ and $f''$ for $\mathsf{G}'$ and $\mathsf{H}$ so that the concatenation of $f'$ and $f''$ has the same output size to $f$. As $f''$ is a valid encoding function for $\mathsf{H}$, its output size is at least $n_2$. In addition, because of the criticality of $\mathsf{G}$, we know that every valid encoding function for $\mathsf{G}'$, including $f'$, needs at least $n_1 + 1$ symbols. Accordingly, the concatenation of $f'$ and $f''$ has an output size of at least $n_1 + n_2 + 1$. Consequently, the output size of the $f$ is at least $n_1 + n_2 + 1$. This means that $\mathsf{G} \cup \mathsf{H}$ is a symmetric critical graph because it cannot support the symmetric rate of $\frac{l}{n_1+n_2}$ after removal of any of its edges.
\mbox{}\hspace*{\fill}\nolinebreak\mbox{$\rule{0.6em}{0.6em}$}
\subsection {Proof of Theorem \ref{thm5}}
\subsubsection{Proof of part (a)}
In the first step of the proof, we will show that  the new graph $\mathsf{G}'$ supports the symmetric rate of $\overline{\bf r} = (\frac{1}{m-1}, \frac{1}{m-1}, \cdots, \frac{1}{m-1})$. In the next step, we will show that this rate will not be achievable if any edge is eliminated. These two steps will clearly prove this theorem.

\emph{(I) Achievability:}
Let $W_i$ denote the message of node $i$, and $W_i \in \mathcal{W}_i = \{0, 1\}$. Consider the encoding function  
$f(W_1, \cdots, W_m) = (f_1, \cdots, f_m) = (W_1 \oplus W_2, \cdots, W_{j-1}\oplus W_j, W_j \oplus W_{j+1} \oplus W_{m+1}, W_{j+1}\oplus W_{j+2}, \cdots, W_{k-1}\oplus W_k, W_k \oplus W_{k+1}  \oplus W_{m+1}, W_{k+1}\oplus W_{k+2}, \cdots, W_{m-1} \oplus W_m, W_m \oplus W_1).$

The $l^{th}$ element of $f$ is the sum of the side information of node $l$ and $W_l$; therefore node $l$ (for $1 \leq l \leq m$) can decode its message. Node $m+1$ can cosider:
\begin{align}
\displaystyle\bigoplus_{l=1}^{i-1}f_l
=&\displaystyle\bigoplus_{l=1, l \neq j}^{i-1}(W_l \oplus W_{l+1}) \oplus (W_j \oplus W_{j+1} \oplus W_{m+1})\\
=&W_1 \oplus W_i \oplus W_{m+1}
\end{align}
Since node $m+1$ has $W_1$ and $W_i$ as side information, it can decode its message with the help public message.

As $f(W_1, \cdots, W_m) \in \{0, 1\}^m$, it shows the achievability of rate $(\frac{1}{m}, \frac{1}{m}, \cdots, \frac{1}{m})$. To prove the achievability of rate $\overline{\bf r}$, notice that for $t \neq j, k$, we have
$$\bigoplus_{l=1}^{m}f_l=0.$$
Therefore, we can only send $(f_2, \cdots, f_m)$; $f_1$ can be omitted from public message and instead recovered from the rest of $f_i$'s. Thus, the  rate $\overline{\bf r}$ is achievable, too.

\emph{(II) Unachievability after Edge Removal:}
To show that $\mathsf{G}'$ is critical, it suffices to prove that after removing any edge of $\mathcal{E'}$, we will need at least $m$ bits of public message. Lemma \ref{lemma:rates-sum} implies that if there exists a subset of length $m$ in a graph which does not contain any cycle, the rate $\overline{\bf r}$ would not be achievable in the graph (otherwise the sum of the rates would be $\frac{m}{m-1}$ which is greater that 1). Thus, it suffices to show that for every $e \in \mathcal{E'}$, there exists a subset of $\mathcal{V'}$, say $A$, of length at least $m$ such that the induced subgraph of $A$ in $\mathsf{G}'-e$ has no directed cycle. 

First, suppose that $e \in \mathcal{E}$. Then we can choose $A = \mathcal{V}$. As $A$ contains exactly $m$ vertices and the induced graph is a directed path which contains no cycle then these edges are critical. For $e = (m+1, 1)$, $A$ can be chosen as $\mathcal{V} \cup \{m+1\} \setminus \{i\}$. Same argument can be made for $e = (m+1, i)$. For $e = (j, m+1)$, $A$ can be chosen as $\mathcal{V} \cup \{m+1\} \setminus \{k\}$. Same argument can be made for $e = (k, m+1)$.

\subsubsection{Proof of part (b)}
We use the same approach from part (a) to show that the rate $\overline{\bf r} = (\frac{1}{m-1}, \frac{1}{m-1}, \cdots, \frac{1}{m-1})$ is achievable in $\mathsf{G}''$, and by removing every edge the rate would not be achievable. 

\emph{(I) Achievability:}
Suppose that the message of node $(u, i)\in \mathcal{V''}$ is $W_{u, i}$. Then, define:
$$W_u = \bigoplus_{i=1}^{n_u}W_{u, i}$$
Similar to part (a), consider the encoding function $f = (f_1, f_2, \cdots, f_m) = (W_1 \oplus W_2, \cdots, W_{j-1}\oplus W_j, W_j \oplus W_{j+1} \oplus W_{m+1}, W_{j+1}\oplus W_{j+2}, \cdots, W_{k-1}\oplus W_k, W_k \oplus W_{k+1}  \oplus W_{m+1}, W_{k+1}\oplus W_{k+2}, \cdots, W_{m-1} \oplus W_m, W_m \oplus W_1)$. Again, for $1\leq l \leq m$ and $t\in [1:n_l]$, the $l^{th}$ element of $f$, $f_l$, is the sum of the side information and $W_{l,t}$. So, these nodes can decode their message. For $t\in [1:n_{m+1}]$, node $(m+1, t)$ can consider:
\begin{align}
\displaystyle\bigoplus_{l=1}^{i-1}f_l
=&\displaystyle\bigoplus_{l=1, l \neq j}^{i-1}(W_l \oplus W_{l+1}) \oplus (W_j \oplus W_{j+1} \oplus W_{m+1})\\
=&W_1 \oplus W_i \oplus W_{m+1}\\
=& \left(\bigoplus_{l=1}^{n_1}W_{1,l}\right) \oplus \left(\bigoplus_{l=1}^{n_i}W_{i,l}\right) \oplus \left(\bigoplus_{l=1}^{n_{m+1}}W_{m+1,l}\right)
\end{align}
By definition of $\mathsf{G}''$, node $(m+1, t)$ knows $W_{1, 1}, \cdots, W_{1, n_1}, W_{i, 1}, \cdots, W_{i, n_i}, W_{m+1, 1}, \cdots, W_{m+1, t-1}, W_{m+1, t+1}$ $, \cdots, W_{m+1, n_{m+1}}$ as side information. Therefore, with the help of public message and its side information $(m+1, t)$ can decode its message.
Additionally, for $t \neq j, k$, we have:
$$\bigoplus_{l=1}^{m}f_l=0.$$
Thus, $f_1$ can be eliminated from the public message (and instead recovered from the rest) and the rate ${\overline{\bf r}}$ would be achieved.

\emph{(II) Unachievability after Edge Removal:}
Now, we want to show that by elimination of any edge in $\mathcal{E''}$, ${\overline{\bf r}}$ will not be achievable anymore. As discussed in part (a), it suffices to show that after removing any edge in $\mathcal{E''}$, there will be $A \subset \mathcal{V''}$ with at least $m$ vertices which does not contain any directed cycle. As we have two different types of edges in $\mathsf{G}''$, we analyze the impact of edge removal on the capacity region in two different cases.

\textbf{case 1)} \emph{$e = ((u, s), (v, t))$ where $u \neq v$}\\
By definition of $\mathsf{G}''$, we have $e' = (u, v) \in \mathcal{E'}$. In part (a), we proved that there exists $A' \subset \mathcal{V'}$ of size $m$ which does not contain any cycle in $\mathsf{G}'-e'$. Now, choose $A =\{(l, 1): s \in A', l \neq u, v\} \cup \{(u, s), (v, t)\}$. If $(x, y)$ and $(x', y') \in A$, then $x, x' \in A'$ and $x \neq x'$. Thus, $(x, y)$ has edge to $(x', y')$ in $\mathsf{G}''$ if and only if $x$ has edge to $x'$ in $\mathsf{G}'$ and because $A'$ has no cycle in $\mathsf{G}'$ then $A$ has no cycle in $\mathsf{G}''$. 

\textbf{case 2)} \emph{e=((u, s), (u, t))}\\
For $1 \leq u \leq m$, choose $A = \{ (u, s), (u, t)\} \cup \{((u+l)$ mod $m, 1), 1\leq l \leq m - 1\}$, and for $u = m + 1$. Choose $A = \{(l, 1): 1\leq l \leq m, l \leq j, k\} \cup \{(u, s), (u, t)\}$. It is straightforward to check that these two sets contain no cycle. 
\mbox{}\hspace*{\fill}\nolinebreak\mbox{$\rule{0.6em}{0.6em}$}
\subsection {Proof of Theorem \ref{thm6}}
\subsubsection{Proof of part (b)}
As mentioned in Remark \ref{rmk:gic-hypergraph}, the underlying digraph of the hypergraph that characterizes a unicast index coding problem equals the directed graph model we used for unicast index coding problem. Thus, the example we offered in part (b) of Theorem \ref{thm2} works for the groupcast scenario too.

\subsubsection{Proof of parts a and c}
Let $\mathsf{H}=(\mathcal{V}, \mathcal{E})$ be a hypergraph related to a groupcast index coding problem. Further, let $\mathsf{G}=(\mathcal{V}, \mathcal{E}_{\mathsf{G}})$ be the underlying directed graph of $\mathsf{H}$. To prove this theorem, it suffices to show that for every edge $e = (W_i, W_j)$ in $\mathsf{G}$ which is not located in any cycles, elimination of $e$ will not change the capacity region. It should be noted that elimination of $e$ from the underlying digraph is interpreted as removing $W_j$ from the side information set of all receivers who intend to find $W_i$. From now on, we will show the hypergraph obtained by the elimination of $e$ from $\mathsf{G}$ using the notation $\mathsf{H} - e$.

Now, we are going to show that for any valid coding function $f$ for side information hypergraph $\mathsf{H}$, there exists a coding function $f'$ which is valid for the groupcast index coding problem introduced by $\mathsf{H} - e$ and its output size equals the output size of $f$. Let  $\mathcal{V}_1$ be the set of vertices that are reachable from $W_j$ and $\mathcal{V}_2=\mathcal{V}-\mathcal{V}_1$ where $\mathcal{V}$ is the vertex set of the graph. Using the assumption that $e$ is not located in any cycles,  one can conclude that $\mathcal{V}$ is partitioned into two parts $\mathcal{V}_1$  and $\mathcal{V}_2$ such that:
\begin{itemize}
\item $W_i \in \mathcal{V}_2$
\item $W_j \in \mathcal{V}_1$
\item $\nexists (u, v) \in \mathcal{E}_{\mathsf{G}}: \quad u \in \mathcal{V}_1 \wedge v \in \mathcal{V}_2$
\end{itemize}
In other words, any receiver who wants to find a message in $\mathcal{V}_1$ does not have any side information about the messages in $\mathcal{V}_2$. This was the only assumption we used in the proof of part (a) of Theorem \ref{thm3new} in order to show that we can find two coding functions $f_1$ and $f_2$ such that all receivers in $\mathcal{V}_i$ be able to find $f$ using $f_1$, $f_2$, and their side information in $\mathcal{V}_i$ and the size of $(f_1, f_2)$ equals the output size of $f$. Hence, using same arguments, $(f_1, f_2)$ is a valid coding function for the groupcast index coding problem introduced by the hypergraph obtained by eliminating edges between $\mathcal{V}_1$ and $\mathcal{V}_2$ (including $e$).
\mbox{}\hspace*{\fill}\nolinebreak\mbox{$\rule{0.6em}{0.6em}$}

\appendices
\section{Two Useful Lemmas}\label{appendixA}
\begin{lemma}[\cite{ZivYitzhakTSTomer}]
\label{lemma:rates-sum}
Assume that $\mathcal{X}$ is a subset of the vertices of a graph $\mathsf{G} = (\mathcal{V}, \mathcal{E})$ which contains no directed cycle. Then in every rate vector ${\overline{\bf r}} = (r_1, \cdots, r_m)$ supported by $\mathsf{G}$ in non-linear asymptotic case, the following holds:\\
\begin{eqnarray}
\sum_{i\in X}r_i \leq 1
\end{eqnarray}
\end{lemma}
Although the lemma above is proved in \cite{ZivYitzhakTSTomer}, we will give a simple operational proof based on graph theory.
\begin{proof}
We construct a new graph $\mathsf{G}' = (\mathcal{V}', \mathcal{E}')$ by contracting the set $\mathcal{X}$ in $\mathsf{G}$.
Strictly speaking, the vertices of $\mathsf{G}'$ include the vertices of $\mathsf{G}$ when we replace all vertices in $\mathcal{X}$ with a single vertex labeled by $\alpha$:
\begin{eqnarray}
\mathcal{V}' = (\mathcal{V}-\mathcal{X})\cup \{\alpha\},
\end{eqnarray}
and edges connected to vertices in $X$ are now connected to $\alpha$ in $\mathsf{G}'$, i.e.
\begin{align*}
\mathcal{E}' =& \{(x,y)\in\mathcal{E}| x\in \mathcal{V}-\mathcal{X}, y\in \mathcal{V}-\mathcal{X}\}
\\&\cup \{(y, \alpha)| y \in \mathcal{V}- \mathcal{X}, \exists x\in \mathcal{X}: (y,x) \in \mathcal{E}\} 
\\&\cup \{(\alpha,y)| y \in \mathcal{V} - \mathcal{X}, \exists x\in \mathcal{X}: (x,y)\in \mathcal{E}\}.
\end{align*}
We prove that if we use the same coding scheme of $\mathsf{G}$ for $\mathsf{G}'$, the node $\alpha$ can decode all messages belonging to the vertices in $\mathcal{X}$. Since the set $\mathcal{X}$ does not contain any cycle, we can order the elements of $\mathcal{X}$ as
$$\mathcal{X}=\{x_1, \cdots, x_t\}$$ such that vertices have only edges to vertices with a higher index, i.e. an edge from vertex $x_i$ to vertex $x_j$ may only exist when $i < j$. Now the vertex $
\alpha$ in $\mathsf{G}'$ can decode the message of $x_t$ due to the fact that $x_t$ is the last element of the order, and therefore, it does not know the messages of the other vertices in $\mathcal{X}$. So, $\alpha$ has all side information of $x_t$ and can decode its message. Next, $x_{t-1}$ can have only the message of $x_t$ from the messages of the vertices in $X$, which has been decoded by now. Thus, $\alpha$ in $\mathsf{G}'$ can decode the message of $x_{t-1}$, too, and this process goes on. Therefore, we can prove by induction that $\alpha$ can obtain all the messages of the vertices in $\mathcal{X}$. In this coding scheme the rate of vertex $\alpha$ equals to:
\begin{eqnarray}
\sum_{i\in \mathcal{X}}r_i
\end{eqnarray}
and by considering the fact that the rate of each vertex cannot be more than $1$, we get our desired result. 
\end{proof}
\begin{lemma}[Tur\'{a}n]
\label{lemma:Turan}
A bidirectional $m$-vertex graph $\mathsf{G}$ that contains no clique of size $k + 1$ has at most $e(m, k)$ edges. Furthermore, the only graph (up to isomorphism) which satisfies the aforementioned condition is $T(m, k)$.
\end{lemma}
The above lemma is known as the Tur\'{a}n Theorem, and its proof can be found in many graph theory books such as \cite[Thm. 5.2.9]{Douglas}.\\

\section{C++ Code for the proof of Theorem \ref{thm2} part (b) }\label{appendixB}

\lstset{
language=C++,
numbers=left,
stepnumber=1,
numbersep=5pt,
backgroundcolor=\color{white},
showspaces=false,
showstringspaces=false,
showtabs=false,
tabsize=2,
captionpos=b,
breaklines=true,
breakatwhitespace=true,
title=\lstname,
}
\begin{lstlisting}
#include <iostream>
#include <vector>
using namespace std;
const int N = 5, M = 5;
vector <int> goodMessages;
int ind[M + 1] = {-1, 0, 0, 0, 0, 0};
bool isGood(int x){ //check whether a combination is good or not.
	for (int I=0; I<N; I++){
		int l = (I + 1) % N, r = (I + N - 1) % N;
		if ((x >> I & 1) == 1 && (x >> l & 1) == 0 && (x >> r & 1) == 0)
			return false;
	}
	return true;
}
bool check(int depth){
	if (depth == M)
		return true;
	for (int I=ind[depth]+1; I<(int)goodMessages.size(); I++){//check for the next combination
		bool flag = true;
		ind[depth + 1] = I;
		for (int K=1; K<=depth; K++)
			flag &= isGood(goodMessages[I] ^ goodMessages[ind[K]]);
		if (flag && check(depth + 1)) //if the XORs with all previous combinations are good then this combination will be added
			return true;
	}	
	return false;
}
int main(){
	for (int mask=1; mask<1<<N; mask++) //find all good combinations
		if (isGood(mask))
			goodMessages.push_back(mask);
	if (check(0))
		cout << "There exists such sets." << endl;
	else
		cout << "Such sets have not been found." << endl;
	return 0;
}
\end{lstlisting}

\section {All symmetric rate critical graphs on 5 nodes}\label{appendixC}
This section provides all symmetric rate critical graphs on 5 nodes using the list given on the website of Young-Han Kim\cite{YHKpersonalwebsite}. There are a total of 9608 graphs listed on the website, among which 32 are critical, appearing from the next page.
\newpage
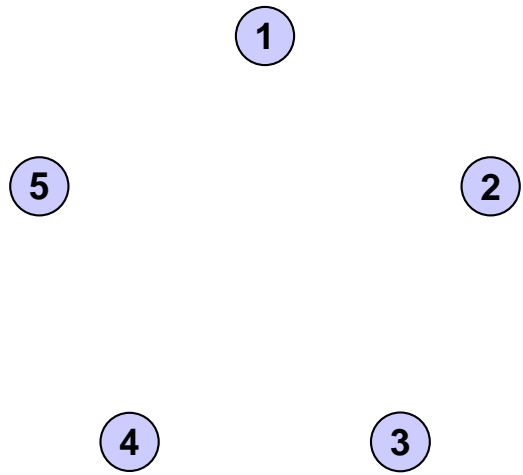
\begin{figure}\centering
\begin{tikzpicture}
[->, >=stealth', shorten >=1pt,auto, node distance=2.5cm, thick, main node/.style={circle,fill=blue!20,draw,font=\sffamily\Large\bfseries}]
\node[main node] (1)at(0,3){1};
\node[main node] (2)at(3,1){2};
\node[main node] (3)at(1.8,-2.4){3};
\node[main node] (4)at(-1.8,-2.4){4};
\node[main node] (5)at(-3,1){5};
\path[every node/.style={font=\sffamily\small}]
;
\end{tikzpicture}
\caption{$\beta$ = 5} \label{fig:2}
\end{figure}

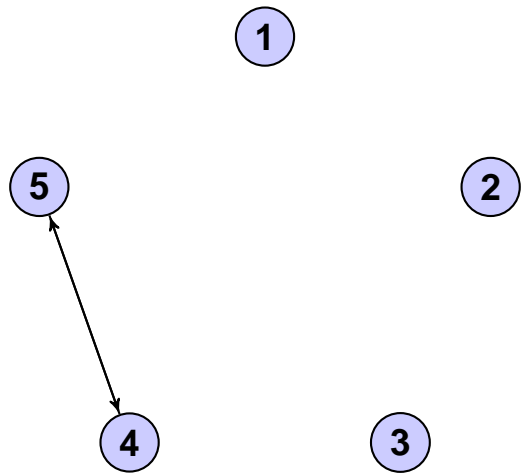
\begin{figure}\centering
\begin{tikzpicture}
[->, >=stealth', shorten >=1pt,auto, node distance=2.5cm, thick, main node/.style={circle,fill=blue!20,draw,font=\sffamily\Large\bfseries}]
\node[main node] (1)at(0,3){1};
\node[main node] (2)at(3,1){2};
\node[main node] (3)at(1.8,-2.4){3};
\node[main node] (4)at(-1.8,-2.4){4};
\node[main node] (5)at(-3,1){5};
\path[every node/.style={font=\sffamily\small}]
(4) edge node {}(5)
(5) edge node {}(4)
;
\end{tikzpicture}
\centering
\caption{$\beta = 4$} \label{fig:5}

\end{figure}
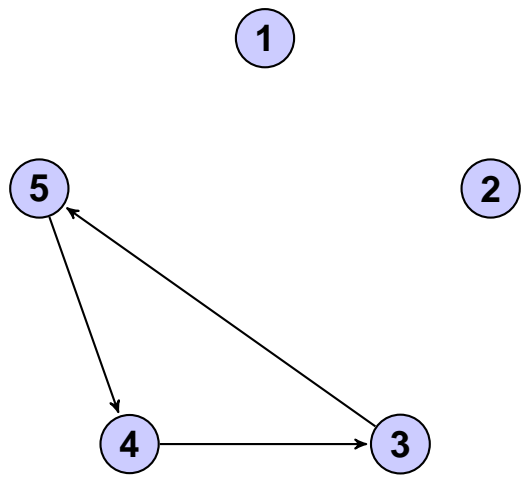
\begin{figure}\centering
\begin{tikzpicture}
[->, >=stealth', shorten >=1pt,auto, node distance=2.5cm, thick, main node/.style={circle,fill=blue!20,draw,font=\sffamily\Large\bfseries}]
\node[main node] (1)at(0,3){1};
\node[main node] (2)at(3,1){2};
\node[main node] (3)at(1.8,-2.4){3};
\node[main node] (4)at(-1.8,-2.4){4};
\node[main node] (5)at(-3,1){5};
\path[every node/.style={font=\sffamily\small}]
(3) edge node {}(5)
(4) edge node {}(3)
(5) edge node {}(4)
;
\end{tikzpicture}
\caption{$\beta$ = 4} \label{fig:18}
\end{figure}
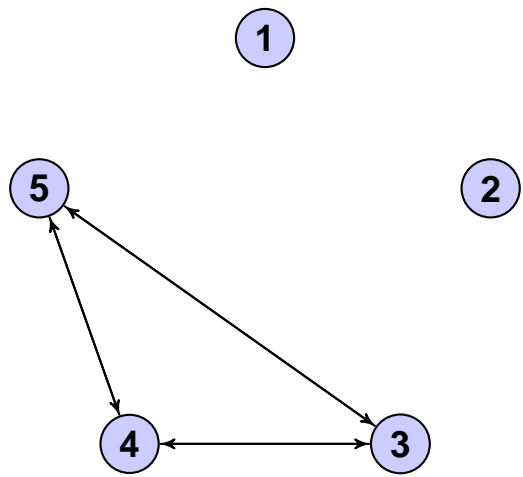
\begin{figure}\centering
\begin{tikzpicture}
[->, >=stealth', shorten >=1pt,auto, node distance=2.5cm, thick, main node/.style={circle,fill=blue!20,draw,font=\sffamily\Large\bfseries}]
\node[main node] (1)at(0,3){1};
\node[main node] (2)at(3,1){2};
\node[main node] (3)at(1.8,-2.4){3};
\node[main node] (4)at(-1.8,-2.4){4};
\node[main node] (5)at(-3,1){5};
\path[every node/.style={font=\sffamily\small}]
(3) edge node {}(4)
edge node {}(5)
(4) edge node {}(3)
edge node {}(5)
(5) edge node {}(3)
edge node {}(4)
;
\end{tikzpicture}
\caption{$\beta$ = 3} \label{fig:279}
\end{figure}
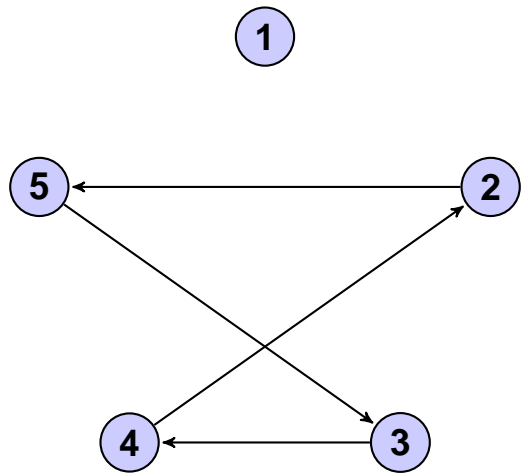
\begin{figure}\centering
\begin{tikzpicture}
[->, >=stealth', shorten >=1pt,auto, node distance=2.5cm, thick, main node/.style={circle,fill=blue!20,draw,font=\sffamily\Large\bfseries}]
\node[main node] (1)at(0,3){1};
\node[main node] (2)at(3,1){2};
\node[main node] (3)at(1.8,-2.4){3};
\node[main node] (4)at(-1.8,-2.4){4};
\node[main node] (5)at(-3,1){5};
\path[every node/.style={font=\sffamily\small}]
(2) edge node {}(5)
(3) edge node {}(4)
(4) edge node {}(2)
(5) edge node {}(3)
;
\end{tikzpicture}
\caption{$\beta$ = 4} \label{fig:78}
\end{figure}
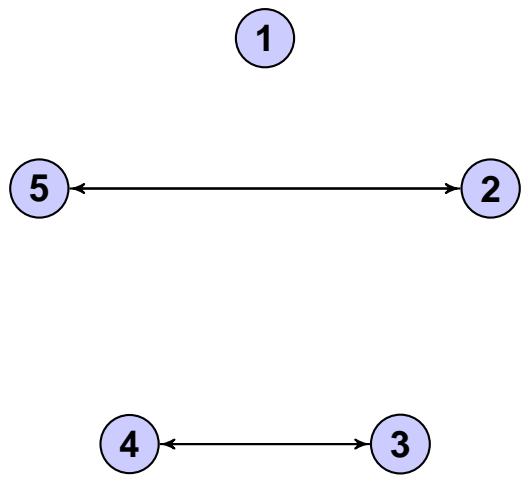
\begin{figure}\centering
\begin{tikzpicture}
[->, >=stealth', shorten >=1pt,auto, node distance=2.5cm, thick, main node/.style={circle,fill=blue!20,draw,font=\sffamily\Large\bfseries}]
\node[main node] (1)at(0,3){1};
\node[main node] (2)at(3,1){2};
\node[main node] (3)at(1.8,-2.4){3};
\node[main node] (4)at(-1.8,-2.4){4};
\node[main node] (5)at(-3,1){5};
\path[every node/.style={font=\sffamily\small}]
(2) edge node {}(5)
(3) edge node {}(4)
(4) edge node {}(3)
(5) edge node {}(2)
;
\end{tikzpicture}
\caption{$\beta$ = 3} \label{fig:76}
\end{figure}
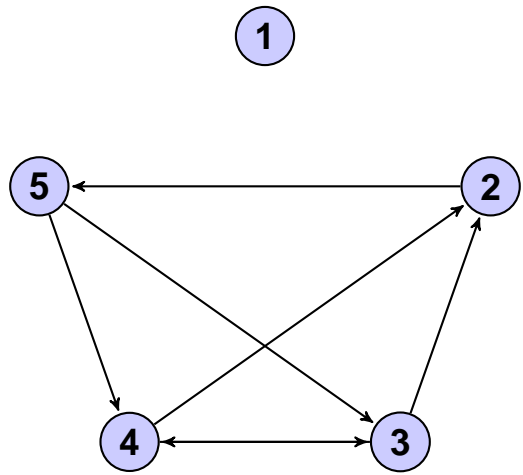
\begin{figure}\centering
\begin{tikzpicture}
[->, >=stealth', shorten >=1pt,auto, node distance=2.5cm, thick, main node/.style={circle,fill=blue!20,draw,font=\sffamily\Large\bfseries}]
\node[main node] (1)at(0,3){1};
\node[main node] (2)at(3,1){2};
\node[main node] (3)at(1.8,-2.4){3};
\node[main node] (4)at(-1.8,-2.4){4};
\node[main node] (5)at(-3,1){5};
\path[every node/.style={font=\sffamily\small}]
(2) edge node {}(5)
(3) edge node {}(2)
edge node {}(4)
(4) edge node {}(2)
edge node {}(3)
(5) edge node {}(3)
edge node {}(4)
;
\end{tikzpicture}
\caption{$\beta$ = 3} \label{fig:893}
\end{figure}
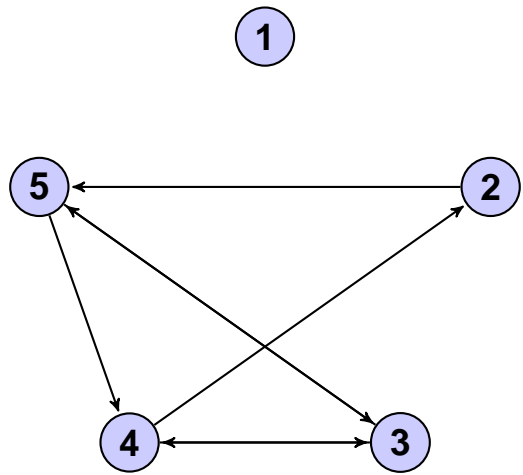
\begin{figure}\centering
\begin{tikzpicture}
[->, >=stealth', shorten >=1pt,auto, node distance=2.5cm, thick, main node/.style={circle,fill=blue!20,draw,font=\sffamily\Large\bfseries}]
\node[main node] (1)at(0,3){1};
\node[main node] (2)at(3,1){2};
\node[main node] (3)at(1.8,-2.4){3};
\node[main node] (4)at(-1.8,-2.4){4};
\node[main node] (5)at(-3,1){5};
\path[every node/.style={font=\sffamily\small}]
(2) edge node {}(5)
(3) edge node {}(4)
edge node {}(5)
(4) edge node {}(2)
edge node {}(3)
(5) edge node {}(3)
edge node {}(4)
;
\end{tikzpicture}
\caption{$\beta$ = 3} \label{fig:771}
\end{figure}
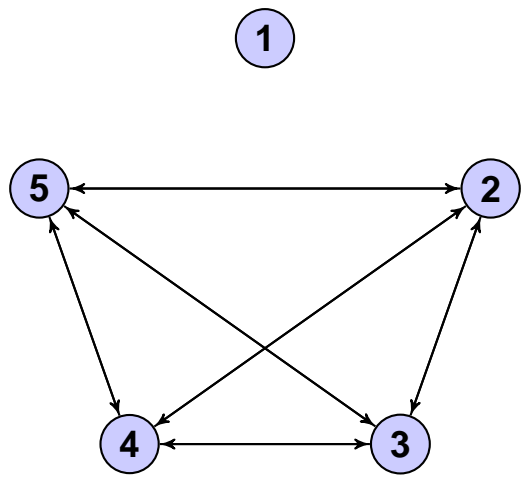
\begin{figure}\centering
\begin{tikzpicture}
[->, >=stealth', shorten >=1pt,auto, node distance=2.5cm, thick, main node/.style={circle,fill=blue!20,draw,font=\sffamily\Large\bfseries}]
\node[main node] (1)at(0,3){1};
\node[main node] (2)at(3,1){2};
\node[main node] (3)at(1.8,-2.4){3};
\node[main node] (4)at(-1.8,-2.4){4};
\node[main node] (5)at(-3,1){5};
\path[every node/.style={font=\sffamily\small}]
(2) edge node {}(3)
edge node {}(4)
edge node {}(5)
(3) edge node {}(2)
edge node {}(4)
edge node {}(5)
(4) edge node {}(2)
edge node {}(3)
edge node {}(5)
(5) edge node {}(2)
edge node {}(3)
edge node {}(4)
;
\end{tikzpicture}
\caption{$\beta$ = 2} \label{fig:7174}
\end{figure}
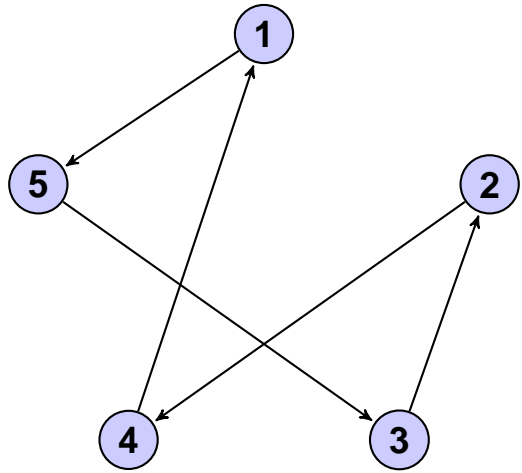
\begin{figure}\centering
\begin{tikzpicture}
[->, >=stealth', shorten >=1pt,auto, node distance=2.5cm, thick, main node/.style={circle,fill=blue!20,draw,font=\sffamily\Large\bfseries}]
\node[main node] (1)at(0,3){1};
\node[main node] (2)at(3,1){2};
\node[main node] (3)at(1.8,-2.4){3};
\node[main node] (4)at(-1.8,-2.4){4};
\node[main node] (5)at(-3,1){5};
\path[every node/.style={font=\sffamily\small}]
(1) edge node {}(5)
(2) edge node {}(4)
(3) edge node {}(2)
(4) edge node {}(1)
(5) edge node {}(3)
;
\end{tikzpicture}
\caption{$\beta$ = 4} \label{fig:239}
\end{figure}
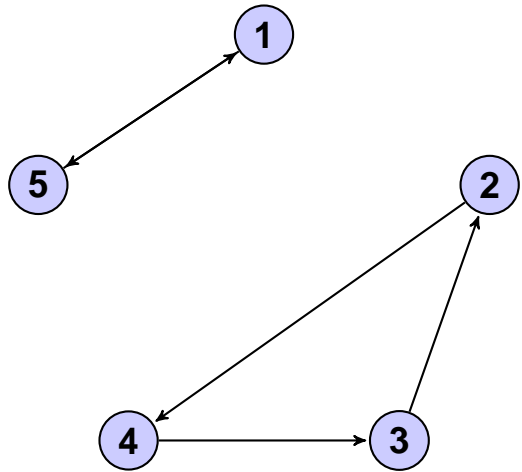
\begin{figure}\centering
\begin{tikzpicture}
[->, >=stealth', shorten >=1pt,auto, node distance=2.5cm, thick, main node/.style={circle,fill=blue!20,draw,font=\sffamily\Large\bfseries}]
\node[main node] (1)at(0,3){1};
\node[main node] (2)at(3,1){2};
\node[main node] (3)at(1.8,-2.4){3};
\node[main node] (4)at(-1.8,-2.4){4};
\node[main node] (5)at(-3,1){5};
\path[every node/.style={font=\sffamily\small}]
(1) edge node {}(5)
(2) edge node {}(4)
(3) edge node {}(2)
(4) edge node {}(3)
(5) edge node {}(1)
;
\end{tikzpicture}
\caption{$\beta$ = 3} \label{fig:238}
\end{figure}
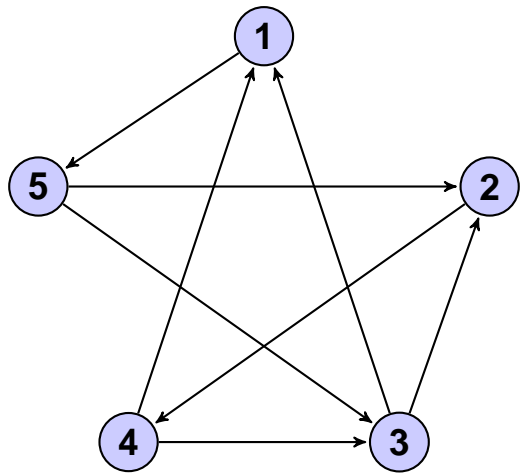
\begin{figure}\centering
\begin{tikzpicture}
[->, >=stealth', shorten >=1pt,auto, node distance=2.5cm, thick, main node/.style={circle,fill=blue!20,draw,font=\sffamily\Large\bfseries}]
\node[main node] (1)at(0,3){1};
\node[main node] (2)at(3,1){2};
\node[main node] (3)at(1.8,-2.4){3};
\node[main node] (4)at(-1.8,-2.4){4};
\node[main node] (5)at(-3,1){5};
\path[every node/.style={font=\sffamily\small}]
(1) edge node {}(5)
(2) edge node {}(4)
(3) edge node {}(1)
edge node {}(2)
(4) edge node {}(1)
edge node {}(3)
(5) edge node {}(2)
edge node {}(3)
;
\end{tikzpicture}
\caption{$\beta$ = 3} \label{fig:2312}
\end{figure}
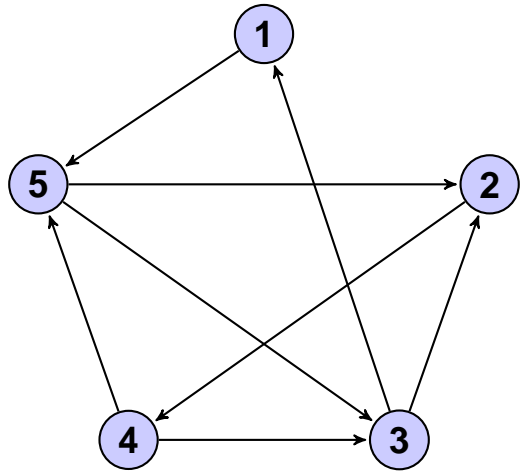
\begin{figure}\centering
\begin{tikzpicture}
[->, >=stealth', shorten >=1pt,auto, node distance=2.5cm, thick, main node/.style={circle,fill=blue!20,draw,font=\sffamily\Large\bfseries}]
\node[main node] (1)at(0,3){1};
\node[main node] (2)at(3,1){2};
\node[main node] (3)at(1.8,-2.4){3};
\node[main node] (4)at(-1.8,-2.4){4};
\node[main node] (5)at(-3,1){5};
\path[every node/.style={font=\sffamily\small}]
(1) edge node {}(5)
(2) edge node {}(4)
(3) edge node {}(1)
edge node {}(2)
(4) edge node {}(3)
edge node {}(5)
(5) edge node {}(2)
edge node {}(3)
;
\end{tikzpicture}
\caption{$\beta$ = 3} \label{fig:2292}
\end{figure}
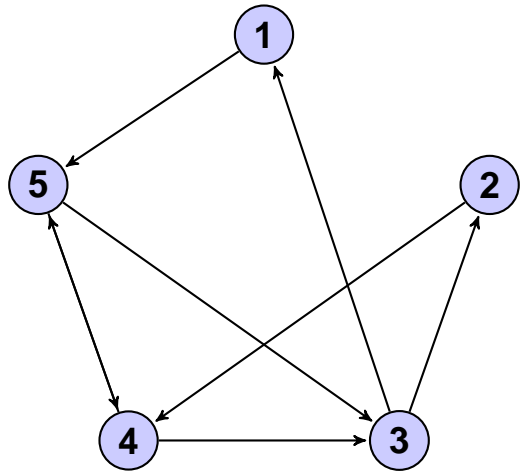
\begin{figure}\centering
\begin{tikzpicture}
[->, >=stealth', shorten >=1pt,auto, node distance=2.5cm, thick, main node/.style={circle,fill=blue!20,draw,font=\sffamily\Large\bfseries}]
\node[main node] (1)at(0,3){1};
\node[main node] (2)at(3,1){2};
\node[main node] (3)at(1.8,-2.4){3};
\node[main node] (4)at(-1.8,-2.4){4};
\node[main node] (5)at(-3,1){5};
\path[every node/.style={font=\sffamily\small}]
(1) edge node {}(5)
(2) edge node {}(4)
(3) edge node {}(1)
edge node {}(2)
(4) edge node {}(3)
edge node {}(5)
(5) edge node {}(3)
edge node {}(4)
;
\end{tikzpicture}
\caption{$\beta$ = 3} \label{fig:2290}
\end{figure}
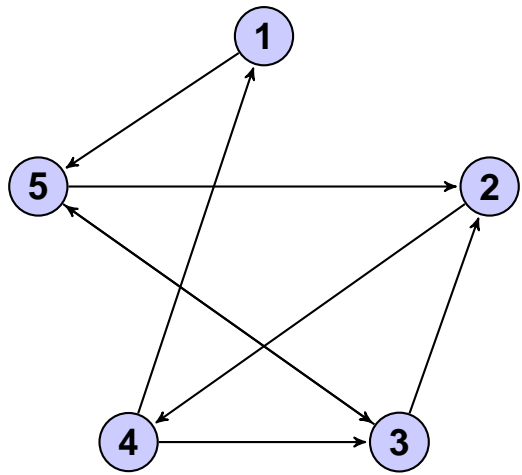
\begin{figure}\centering
\begin{tikzpicture}
[->, >=stealth', shorten >=1pt,auto, node distance=2.5cm, thick, main node/.style={circle,fill=blue!20,draw,font=\sffamily\Large\bfseries}]
\node[main node] (1)at(0,3){1};
\node[main node] (2)at(3,1){2};
\node[main node] (3)at(1.8,-2.4){3};
\node[main node] (4)at(-1.8,-2.4){4};
\node[main node] (5)at(-3,1){5};
\path[every node/.style={font=\sffamily\small}]
(1) edge node {}(5)
(2) edge node {}(4)
(3) edge node {}(2)
edge node {}(5)
(4) edge node {}(1)
edge node {}(3)
(5) edge node {}(2)
edge node {}(3)
;
\end{tikzpicture}
\caption{$\beta$ = 3} \label{fig:2202}
\end{figure}
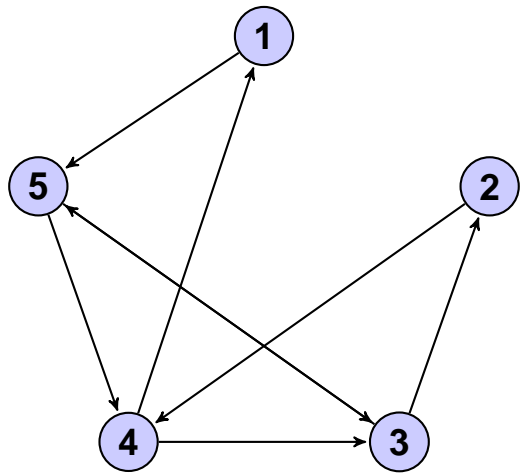
\begin{figure}\centering
\begin{tikzpicture}
[->, >=stealth', shorten >=1pt,auto, node distance=2.5cm, thick, main node/.style={circle,fill=blue!20,draw,font=\sffamily\Large\bfseries}]
\node[main node] (1)at(0,3){1};
\node[main node] (2)at(3,1){2};
\node[main node] (3)at(1.8,-2.4){3};
\node[main node] (4)at(-1.8,-2.4){4};
\node[main node] (5)at(-3,1){5};
\path[every node/.style={font=\sffamily\small}]
(1) edge node {}(5)
(2) edge node {}(4)
(3) edge node {}(2)
edge node {}(5)
(4) edge node {}(1)
edge node {}(3)
(5) edge node {}(3)
edge node {}(4)
;
\end{tikzpicture}
\caption{$\beta$ = 3} \label{fig:2200}
\end{figure}
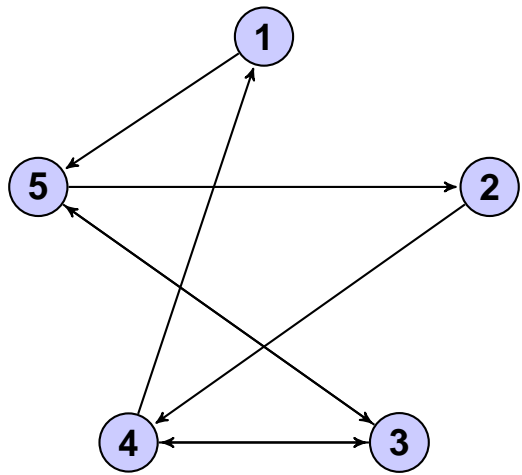
\begin{figure}\centering
\begin{tikzpicture}
[->, >=stealth', shorten >=1pt,auto, node distance=2.5cm, thick, main node/.style={circle,fill=blue!20,draw,font=\sffamily\Large\bfseries}]
\node[main node] (1)at(0,3){1};
\node[main node] (2)at(3,1){2};
\node[main node] (3)at(1.8,-2.4){3};
\node[main node] (4)at(-1.8,-2.4){4};
\node[main node] (5)at(-3,1){5};
\path[every node/.style={font=\sffamily\small}]
(1) edge node {}(5)
(2) edge node {}(4)
(3) edge node {}(4)
edge node {}(5)
(4) edge node {}(1)
edge node {}(3)
(5) edge node {}(2)
edge node {}(3)
;
\end{tikzpicture}
\caption{$\beta$ = 3} \label{fig:2120}
\end{figure}
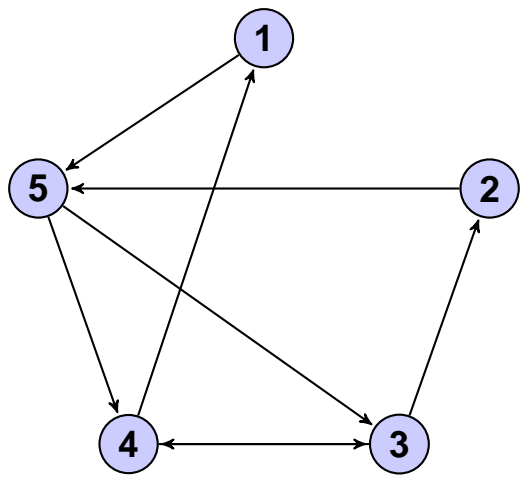
\begin{figure}\centering
\begin{tikzpicture}
[->, >=stealth', shorten >=1pt,auto, node distance=2.5cm, thick, main node/.style={circle,fill=blue!20,draw,font=\sffamily\Large\bfseries}]
\node[main node] (1)at(0,3){1};
\node[main node] (2)at(3,1){2};
\node[main node] (3)at(1.8,-2.4){3};
\node[main node] (4)at(-1.8,-2.4){4};
\node[main node] (5)at(-3,1){5};
\path[every node/.style={font=\sffamily\small}]
(1) edge node {}(5)
(2) edge node {}(5)
(3) edge node {}(2)
edge node {}(4)
(4) edge node {}(1)
edge node {}(3)
(5) edge node {}(3)
edge node {}(4)
;
\end{tikzpicture}
\caption{$\beta$ = 3} \label{fig:2055}
\end{figure}

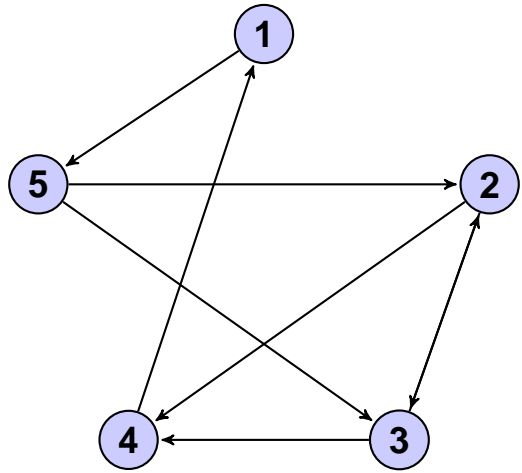
\begin{figure}\centering
\begin{tikzpicture}
[->, >=stealth', shorten >=1pt,auto, node distance=2.5cm, thick, main node/.style={circle,fill=blue!20,draw,font=\sffamily\Large\bfseries}]
\node[main node] (1)at(0,3){1};
\node[main node] (2)at(3,1){2};
\node[main node] (3)at(1.8,-2.4){3};
\node[main node] (4)at(-1.8,-2.4){4};
\node[main node] (5)at(-3,1){5};
\path[every node/.style={font=\sffamily\small}]
(1) edge node {}(5)
(2) edge node {}(3)
edge node {}(4)
(3) edge node {}(2)
edge node {}(4)
(4) edge node {}(1)
(5) edge node {}(2)
edge node {}(3)
;
\end{tikzpicture}
\caption{$\beta$ = 3} \label{fig:2469}
\end{figure}
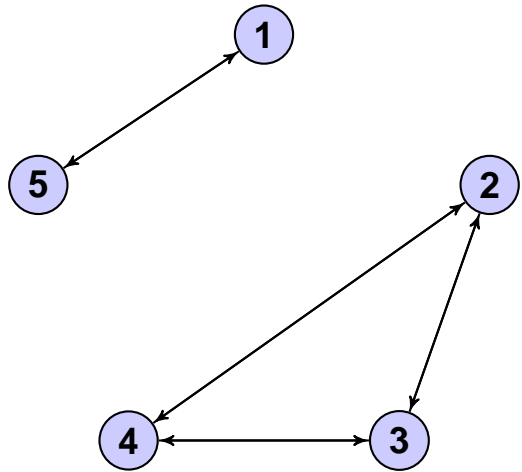
\begin{figure}\centering
\begin{tikzpicture}
[->, >=stealth', shorten >=1pt,auto, node distance=2.5cm, thick, main node/.style={circle,fill=blue!20,draw,font=\sffamily\Large\bfseries}]
\node[main node] (1)at(0,3){1};
\node[main node] (2)at(3,1){2};
\node[main node] (3)at(1.8,-2.4){3};
\node[main node] (4)at(-1.8,-2.4){4};
\node[main node] (5)at(-3,1){5};
\path[every node/.style={font=\sffamily\small}]
(1) edge node {}(5)
(2) edge node {}(3)
edge node {}(4)
(3) edge node {}(2)
edge node {}(4)
(4) edge node {}(2)
edge node {}(3)
(5) edge node {}(1)
;
\end{tikzpicture}
\caption{$\beta$ = 2} \label{fig:2467}
\end{figure}
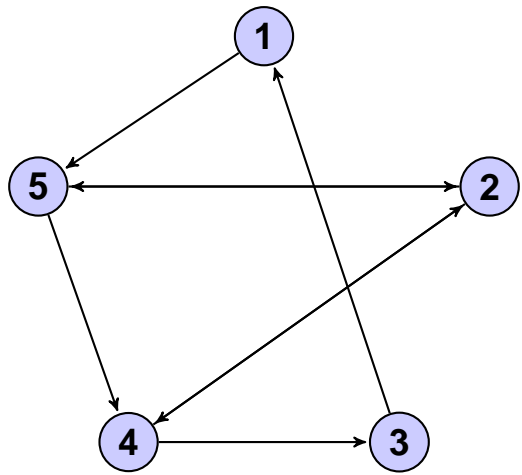
\begin{figure}\centering
\begin{tikzpicture}
[->, >=stealth', shorten >=1pt,auto, node distance=2.5cm, thick, main node/.style={circle,fill=blue!20,draw,font=\sffamily\Large\bfseries}]
\node[main node] (1)at(0,3){1};
\node[main node] (2)at(3,1){2};
\node[main node] (3)at(1.8,-2.4){3};
\node[main node] (4)at(-1.8,-2.4){4};
\node[main node] (5)at(-3,1){5};
\path[every node/.style={font=\sffamily\small}]
(1) edge node {}(5)
(2) edge node {}(4)
edge node {}(5)
(3) edge node {}(1)
(4) edge node {}(2)
edge node {}(3)
(5) edge node {}(2)
edge node {}(4)
;
\end{tikzpicture}
\caption{$\beta$ = 3} \label{fig:2404}
\end{figure}
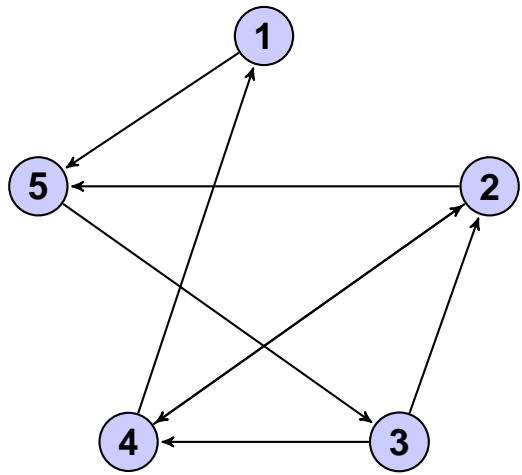
\begin{figure}\centering
\begin{tikzpicture}
[->, >=stealth', shorten >=1pt,auto, node distance=2.5cm, thick, main node/.style={circle,fill=blue!20,draw,font=\sffamily\Large\bfseries}]
\node[main node] (1)at(0,3){1};
\node[main node] (2)at(3,1){2};
\node[main node] (3)at(1.8,-2.4){3};
\node[main node] (4)at(-1.8,-2.4){4};
\node[main node] (5)at(-3,1){5};
\path[every node/.style={font=\sffamily\small}]
(1) edge node {}(5)
(2) edge node {}(4)
edge node {}(5)
(3) edge node {}(2)
edge node {}(4)
(4) edge node {}(1)
edge node {}(2)
(5) edge node {}(3)
;
\end{tikzpicture}
\caption{$\beta$ = 3} \label{fig:2397}
\end{figure}
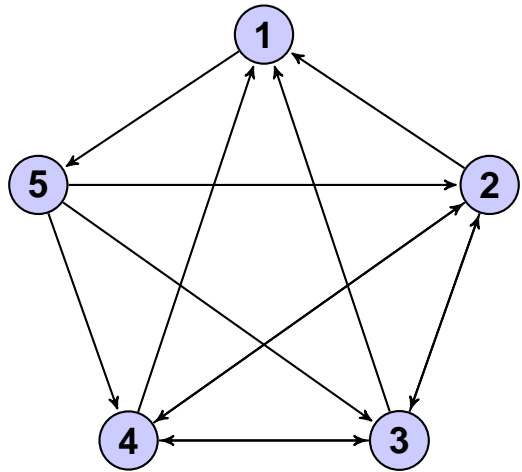
\begin{figure}\centering
\begin{tikzpicture}
[->, >=stealth', shorten >=1pt,auto, node distance=2.5cm, thick, main node/.style={circle,fill=blue!20,draw,font=\sffamily\Large\bfseries}]
\node[main node] (1)at(0,3){1};
\node[main node] (2)at(3,1){2};
\node[main node] (3)at(1.8,-2.4){3};
\node[main node] (4)at(-1.8,-2.4){4};
\node[main node] (5)at(-3,1){5};
\path[every node/.style={font=\sffamily\small}]
(1) edge node {}(5)
(2) edge node {}(1)
edge node {}(3)
edge node {}(4)
(3) edge node {}(1)
edge node {}(2)
edge node {}(4)
(4) edge node {}(1)
edge node {}(2)
edge node {}(3)
(5) edge node {}(2)
edge node {}(3)
edge node {}(4)
;
\end{tikzpicture}
\caption{$\beta$ = 2} \label{fig:8625}
\end{figure}
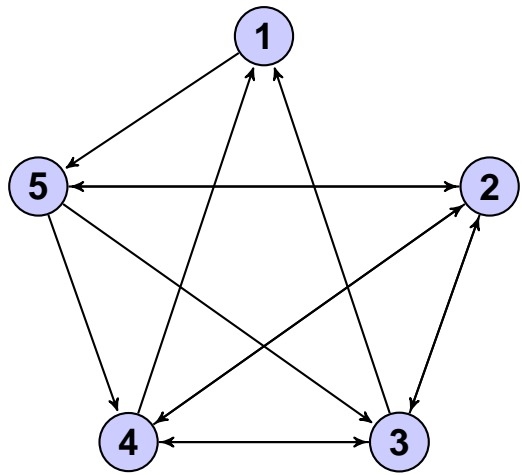
\begin{figure}\centering
\begin{tikzpicture}
[->, >=stealth', shorten >=1pt,auto, node distance=2.5cm, thick, main node/.style={circle,fill=blue!20,draw,font=\sffamily\Large\bfseries}]
\node[main node] (1)at(0,3){1};
\node[main node] (2)at(3,1){2};
\node[main node] (3)at(1.8,-2.4){3};
\node[main node] (4)at(-1.8,-2.4){4};
\node[main node] (5)at(-3,1){5};
\path[every node/.style={font=\sffamily\small}]
(1) edge node {}(5)
(2) edge node {}(3)
edge node {}(4)
edge node {}(5)
(3) edge node {}(1)
edge node {}(2)
edge node {}(4)
(4) edge node {}(1)
edge node {}(2)
edge node {}(3)
(5) edge node {}(2)
edge node {}(3)
edge node {}(4)
;
\end{tikzpicture}
\caption{$\beta$ = 2} \label{fig:8495}
\end{figure}
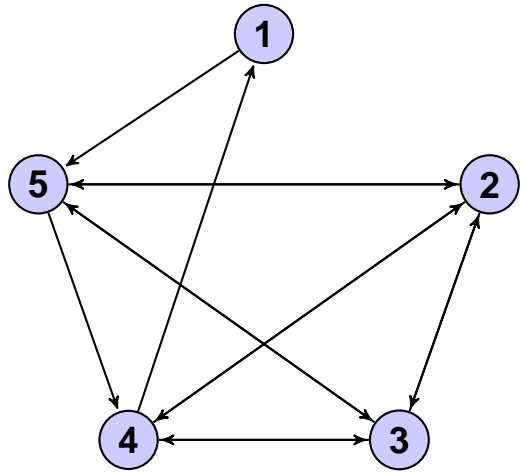
\begin{figure}\centering
\begin{tikzpicture}
[->, >=stealth', shorten >=1pt,auto, node distance=2.5cm, thick, main node/.style={circle,fill=blue!20,draw,font=\sffamily\Large\bfseries}]
\node[main node] (1)at(0,3){1};
\node[main node] (2)at(3,1){2};
\node[main node] (3)at(1.8,-2.4){3};
\node[main node] (4)at(-1.8,-2.4){4};
\node[main node] (5)at(-3,1){5};
\path[every node/.style={font=\sffamily\small}]
(1) edge node {}(5)
(2) edge node {}(3)
edge node {}(4)
edge node {}(5)
(3) edge node {}(2)
edge node {}(4)
edge node {}(5)
(4) edge node {}(1)
edge node {}(2)
edge node {}(3)
(5) edge node {}(2)
edge node {}(3)
edge node {}(4)
;
\end{tikzpicture}
\caption{$\beta$ = 2} \label{fig:8436}
\end{figure}
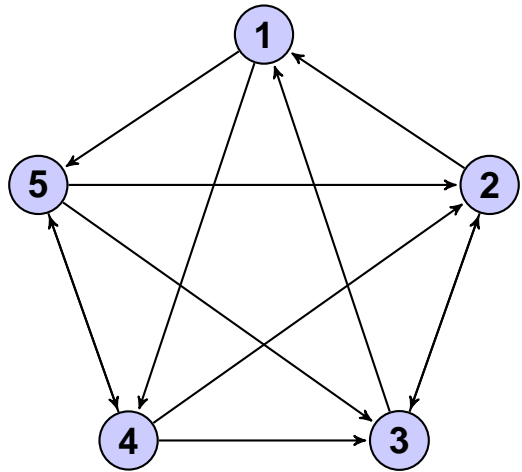
\begin{figure}\centering
\begin{tikzpicture}
[->, >=stealth', shorten >=1pt,auto, node distance=2.5cm, thick, main node/.style={circle,fill=blue!20,draw,font=\sffamily\Large\bfseries}]
\node[main node] (1)at(0,3){1};
\node[main node] (2)at(3,1){2};
\node[main node] (3)at(1.8,-2.4){3};
\node[main node] (4)at(-1.8,-2.4){4};
\node[main node] (5)at(-3,1){5};
\path[every node/.style={font=\sffamily\small}]
(1) edge node {}(4)
edge node {}(5)
(2) edge node {}(1)
edge node {}(3)
(3) edge node {}(1)
edge node {}(2)
(4) edge node {}(2)
edge node {}(3)
edge node {}(5)
(5) edge node {}(2)
edge node {}(3)
edge node {}(4)
;
\end{tikzpicture}
\caption{$\beta$ = 2} \label{fig:8285}
\end{figure}
\begin{figure}\centering
\begin{tikzpicture}
[->, >=stealth', shorten >=1pt,auto, node distance=2.5cm, thick, main node/.style={circle,fill=blue!20,draw,font=\sffamily\Large\bfseries}]
\node[main node] (1)at(0,3){1};
\node[main node] (2)at(3,1){2};
\node[main node] (3)at(1.8,-2.4){3};
\node[main node] (4)at(-1.8,-2.4){4};
\node[main node] (5)at(-3,1){5};
\path[every node/.style={font=\sffamily\small}]
(1) edge node {}(4)
edge node {}(5)
(2) edge node {}(3)
edge node {}(5)
(3) edge node {}(2)
edge node {}(4)
(4) edge node {}(1)
edge node {}(3)
(5) edge node {}(1)
edge node {}(2)
;
\end{tikzpicture}
\caption{$\beta$ = 2.5} \label{fig:5626}
\end{figure}
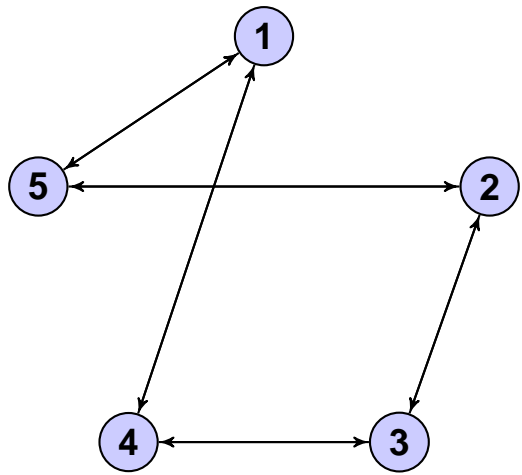
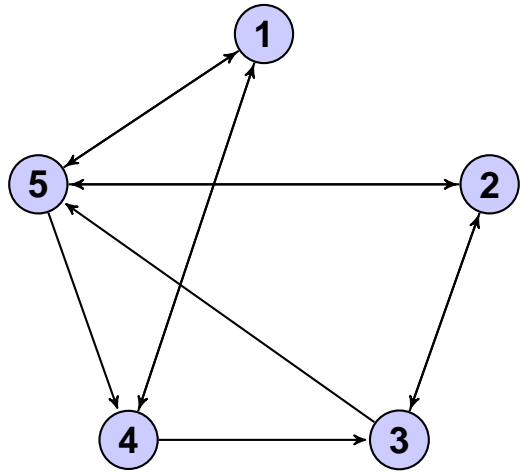
\begin{figure}\centering
\begin{tikzpicture}
[->, >=stealth', shorten >=1pt,auto, node distance=2.5cm, thick, main node/.style={circle,fill=blue!20,draw,font=\sffamily\Large\bfseries}]
\node[main node] (1)at(0,3){1};
\node[main node] (2)at(3,1){2};
\node[main node] (3)at(1.8,-2.4){3};
\node[main node] (4)at(-1.8,-2.4){4};
\node[main node] (5)at(-3,1){5};
\path[every node/.style={font=\sffamily\small}]
(1) edge node {}(4)
edge node {}(5)
(2) edge node {}(3)
edge node {}(5)
(3) edge node {}(2)
edge node {}(5)
(4) edge node {}(1)
edge node {}(3)
(5) edge node {}(1)
edge node {}(2)
edge node {}(4)
;
\end{tikzpicture}
\caption{$\beta$ = 2.5} \label{fig:7016}
\end{figure}
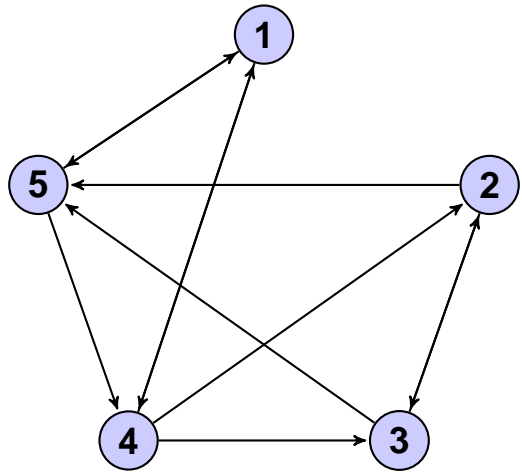
\begin{figure}\centering
\begin{tikzpicture}
[->, >=stealth', shorten >=1pt,auto, node distance=2.5cm, thick, main node/.style={circle,fill=blue!20,draw,font=\sffamily\Large\bfseries}]
\node[main node] (1)at(0,3){1};
\node[main node] (2)at(3,1){2};
\node[main node] (3)at(1.8,-2.4){3};
\node[main node] (4)at(-1.8,-2.4){4};
\node[main node] (5)at(-3,1){5};
\path[every node/.style={font=\sffamily\small}]
(1) edge node {}(4)
edge node {}(5)
(2) edge node {}(3)
edge node {}(5)
(3) edge node {}(2)
edge node {}(5)
(4) edge node {}(1)
edge node {}(2)
edge node {}(3)
(5) edge node {}(1)
edge node {}(4)
;
\end{tikzpicture}
\caption{$\beta$ = 2} \label{fig:7026}
\end{figure}
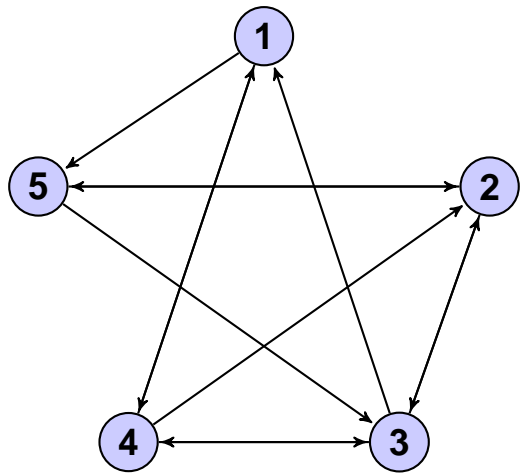
\begin{figure}\centering
\begin{tikzpicture}
[->, >=stealth', shorten >=1pt,auto, node distance=2.5cm, thick, main node/.style={circle,fill=blue!20,draw,font=\sffamily\Large\bfseries}]
\node[main node] (1)at(0,3){1};
\node[main node] (2)at(3,1){2};
\node[main node] (3)at(1.8,-2.4){3};
\node[main node] (4)at(-1.8,-2.4){4};
\node[main node] (5)at(-3,1){5};
\path[every node/.style={font=\sffamily\small}]
(1) edge node {}(4)
edge node {}(5)
(2) edge node {}(3)
edge node {}(5)
(3) edge node {}(1)
edge node {}(2)
edge node {}(4)
(4) edge node {}(1)
edge node {}(2)
edge node {}(3)
(5) edge node {}(2)
edge node {}(3)
;
\end{tikzpicture}
\caption{$\beta$ = 2} \label{fig:8161}
\end{figure}
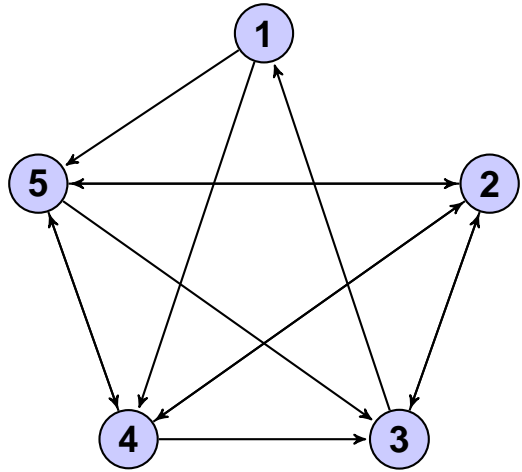
\begin{figure}\centering
\begin{tikzpicture}
[->, >=stealth', shorten >=1pt,auto, node distance=2.5cm, thick, main node/.style={circle,fill=blue!20,draw,font=\sffamily\Large\bfseries}]
\node[main node] (1)at(0,3){1};
\node[main node] (2)at(3,1){2};
\node[main node] (3)at(1.8,-2.4){3};
\node[main node] (4)at(-1.8,-2.4){4};
\node[main node] (5)at(-3,1){5};
\path[every node/.style={font=\sffamily\small}]
(1) edge node {}(4)
edge node {}(5)
(2) edge node {}(3)
edge node {}(4)
edge node {}(5)
(3) edge node {}(1)
edge node {}(2)
(4) edge node {}(2)
edge node {}(3)
edge node {}(5)
(5) edge node {}(2)
edge node {}(3)
edge node {}(4)
;
\end{tikzpicture}
\caption{$\beta$ = 2} \label{fig:8847}
\end{figure}
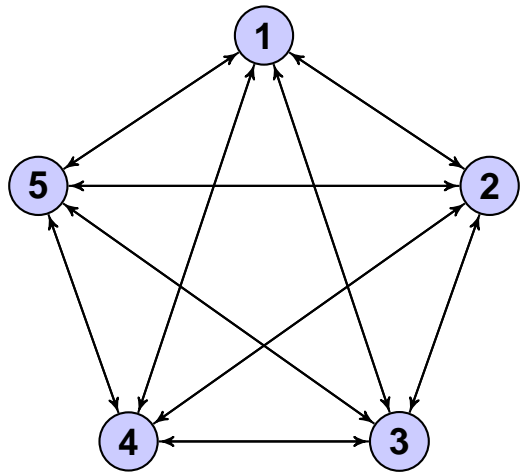
\begin{figure}\centering
\begin{tikzpicture}
[->, >=stealth', shorten >=1pt,auto, node distance=2.5cm, thick, main node/.style={circle,fill=blue!20,draw,font=\sffamily\Large\bfseries}]
\node[main node] (1)at(0,3){1};
\node[main node] (2)at(3,1){2};
\node[main node] (3)at(1.8,-2.4){3};
\node[main node] (4)at(-1.8,-2.4){4};
\node[main node] (5)at(-3,1){5};
\path[every node/.style={font=\sffamily\small}]
(1) edge node {}(2)
edge node {}(3)
edge node {}(4)
edge node {}(5)
(2) edge node {}(1)
edge node {}(3)
edge node {}(4)
edge node {}(5)
(3) edge node {}(1)
edge node {}(2)
edge node {}(4)
edge node {}(5)
(4) edge node {}(1)
edge node {}(2)
edge node {}(3)
edge node {}(5)
(5) edge node {}(1)
edge node {}(2)
edge node {}(3)
edge node {}(4)
;
\end{tikzpicture}
\caption{$\beta$ = 1} \label{fig:9609}
\end{figure}



\ifCLASSOPTIONcaptionsoff
  \newpage
\fi


\begin{thebibliography}{16}
\bibitem{YitzhakTomer} Y. Birk and T. Kol, \emph{Informed-source coding-on-demand (ISCOD) over broadcast channels,} in Proc. 17th Ann. IEEE Int. Conf. Comput. Commun. (INFOCOM), San Francisco, CA, Mar. 1998, pp. 1257-1264.
\bibitem{SonHoangVitalyYeowMeng} S. H. Dau, V. Skachek, and Y. M. Chee, \emph{Secure index coding with side information,} arXiv preprint arXiv:1011.5566, 2010.
\bibitem{AnnaRobertEyal} A. Blasiak, R. Kleinberg, and E. Lubetzky, \emph{Lexicographic Products and the Power of Non-Linear Network Coding,} in Proc. of the 52nd Annual IEEE Symposium on Foundations of Computer Science (FOCS), 2011, pp. 609-618.
\bibitem{FatemehBerndYoungHanErenLele} F. Arbabjolfaei, B. Bandemer, Y. Kim, E. Sasoglu, and L. Wang, \emph{On the Capacity Region for Index Coding,} arXiv preprint arXiv:1302.1601v2.pdf, 2013.
\bibitem{HuaSyedAli}H. Sun and S. A. Jafar, \emph{Index coding capacity: How far can one go with only Shannon inequalities?,} arXiv preprint arXiv:1303.7000.pdf, 2013.
\bibitem{KarthikeyanAlexandrosMichael} K. Shanmugam, A. G. Dimakis, and M. Langberg, \emph{Local graph coloring and index coding,} in IEEE International Symposium on Information Theory (ISIT), 2013, pp. 1152-1156.
\bibitem{SalimAlexCostas}S. El Rouayheb, A. Sprintson, and C. Georghiades, \emph{On the relation between the Index Coding and the Network Coding problems,} IEEE International Symposium on Information Theory (ISIT), 2008, pp. 1823-1827. 
\bibitem{MichelleSalimMichelle} M. Effros, S. Rouayheb, and M. Langberg, \emph{An equivalence between network coding and index coding,} arXiv preprint arXiv:1211.6660, 2012.
\bibitem{HamedViveckSyedAli} H. Maleki, V. Cadambe, and S. Jafar, \emph{Index coding: An interference alignment perspective,} in Information Theory Proceedings (ISIT), 2012 IEEE International Symposium on. IEEE, 2012, pp. 2236-2240.
\bibitem{EyalUri} E. Lubetzky and U. Stav, \emph{Non-linear index coding outperforming the linear optimum,} in Proc. of the 48th Annual IEEE Symposium on Foundations of Computer Science (FOCS), 2007, pp. 161-167.
\bibitem{ZivYitzhakTSTomer} Z. Bar-Yossef, Y. Birk, T. S. Jayram, and T. Kol, \emph{Index coding with side information,} in 47th Annual IEEE Symposium on Foundations of Computer Science, (FOCS), 2006, pp. 197-206.
\bibitem{HavivLangberg} M. Langberg and M. Effros, \emph{Network coding: Is zero error always possible?} in 49th Annual Allerton Conference on Communications, Control and Computing, 2011, pp. 1478-1485.
\bibitem{TraceyMichelleShirin} T. Ho, M. Effros, and S. Jalali, \emph{On equivalences between network topologies,} in 48th Annual Allerton Conference on Communication, Control, and Computing, 2010.
\bibitem{ShirinMichelleTracey}S. Jalali, M. Effros, and T. Ho, \emph{On the impact of a single edge on the network coding capacity,} in Information Theory and Applications Workshop (ITA), 2011.
\bibitem{Douglas} D.B.West, \emph{Introduction to Graph Theory (2nd Edition),} Prentice Hall, 2001.
\bibitem{AbbasYoungHan} A. El Gamal and Y.-H. Kim, \emph{Network Information Theory,} Cambridge University Press, 2011.
\bibitem{SyedAli} Jafar, Syed A. \emph{Topological Interference Management through Index Coding.} arXiv preprint arXiv:1301.3106, 2013.
\bibitem{NeelySaberZhang} Neely, Michael J., Arash Saber Tehrani, and Zhen Zhang. \emph{Dynamic index coding for wireless broadcast networks.} In INFOCOM, 2012 Proceedings IEEE, pp. 316-324. IEEE, 2012.
\bibitem{percolation-paper}
D.K. Arrowsmith, J. W. Essam, \emph{Percolation theory on directed graphs,} J. Mathematical Phys. 18 (1977), no. 2, 235-238.
\bibitem{YHKpersonalwebsite} Young-Han Kim, \emph{Index Coding}, available at \url{http://circuit.ucsd.edu/~yhk/indexcoding.html}.

\end{thebibliography}
\end{document}